\documentclass[a4paper,twocolumn,10pt,accepted=2025-06-02]{quantumarticle}
\pdfoutput=1

\usepackage[utf8]{inputenc}
\usepackage[english]{babel}
\usepackage[T1]{fontenc}
\usepackage{amsmath,mathtools}
\makeatletter
\newcommand{\pushright}[1]{\ifmeasuring@#1\else\omit\hfill$\displaystyle#1$\fi\ignorespaces}
\newcommand{\pushleft}[1]{\ifmeasuring@#1\else\omit$\displaystyle#1$\hfill\fi\ignorespaces}
\makeatother

\allowdisplaybreaks

\usepackage{tikz}
\usepackage{lipsum}
\usepackage[numbers,sort&compress,square]{natbib}
\usepackage{amssymb}
\usepackage[normalem]{ulem}
\usepackage{color}
\usepackage{mathrsfs}
\usepackage{amsbsy}
\usepackage{amsthm}
\usepackage{array}
\newcolumntype{?}{!{\vrule width 1.5pt}}

\usepackage{balance}
\usepackage{thmtools,thm-restate}
\usepackage{graphicx}
\usepackage{tikz, circuitikz}

\usepackage{comment}

\usepackage{bbm}
\usepackage{bm}
\usepackage{epsfig}
\usepackage{xfrac}
\usepackage{xcolor}
\definecolor{ginger}{rgb}{0.69, 0.4, 0.0}

\usepackage{enumerate}
\usepackage[shortlabels]{enumitem}
\usepackage{graphicx}
\usepackage{bm}
\usepackage{hyperref}
\usepackage{subfigure}
\hypersetup{
colorlinks=true,
linkcolor=blue,
filecolor=blue,
citecolor=blue,  
urlcolor=blue,
}

\newtheorem{property}{Property}

\newtheorem{theorem}{Theorem}

\newtheorem{lemma}{Lemma}
\newtheorem{manualtheoreminner}{Property}
\newenvironment{manualtheorem}[1]{%
  \renewcommand\themanualtheoreminner{#1}%
  \manualtheoreminner
}{\endmanualtheoreminner}

\DeclareMathOperator{\wt}{wt}

\usepackage{stackengine}

\usepackage{braket}
\newcommand{\abs}[1]{\left| #1 \right|}
\newcommand{\norm}[1]{\left\| #1 \right\|} 
\newcommand{\avg}{{\mathrm{avg}}} 

\newcommand{\eff}{\mathrm{eff}} 

\newcommand{\tl}{{\textsc{l}}}

\newcommand{\sfS}{{\mathsf{S}}}

\newcommand{\sfG}{{\mathsf{G}}}
\newcommand{\mS}{{\mathcal{S}}}

\newcommand{\frakF}{{\mathfrak{F}}}

\newcommand{\frakR}{{\mathfrak{R}}}
\newcommand{\frakG}{{\mathfrak{G}}}

\newcommand{\mP}{{\mathcal{P}}}

\newcommand{\scrS}{{\mathscr{S}}}
\newcommand{\scrN}{{\mathscr{N}}}
\newcommand{\scrF}{{\mathscr{F}}}

\newcommand{\bR}{{\mathbb{R}}}

\newcommand{\mbv}{\bm{v}}

\newcommand{\opt}{{\mathrm{opt}}}
\renewcommand{\epsilon}{\varepsilon}
\newcommand{\appropto}{\mathrel{\vcenter{
  \offinterlineskip\halign{\hfil$##$\cr
    \propto\cr\noalign{\kern2pt}\sim\cr\noalign{\kern-2pt}}}}}

\definecolor{fluorescentpink}{rgb}{1.0, 0.08, 0.58}

\let\baraccent=\= 
\renewcommand{\=}[1]{\stackrel{#1}{=}} 

\newcommand{\thmref}[1]{\hyperref[#1]{Theorem~\ref{#1}}}
\newcommand{\lemmaref}[1]{\hyperref[#1]{Lemma~\ref{#1}}}
\newcommand{\propref}[1]{\hyperref[#1]{Property~\ref{#1}}}
\newcommand{\corollaryref}[1]{\hyperref[#1]{Corollary~\ref{#1}}}
\newcommand{\figref}[1]{\hyperref[#1]{Fig.~\ref{#1}}}
\newcommand{\tabref}[1]{\hyperref[#1]{Table~\ref{#1}}}
\newcommand{\figaref}[1]{\hyperref[#1]{Fig.~\ref{#1}(a)}}
\newcommand{\figbref}[1]{\hyperref[#1]{Fig.~\ref{#1}(b)}}
\newcommand{\figcref}[1]{\hyperref[#1]{Fig.~\ref{#1}(c)}}
\newcommand{\figdref}[1]{\hyperref[#1]{Fig.~\ref{#1}(d)}}
\newcommand{\figeref}[1]{\hyperref[#1]{Fig.~\ref{#1}(e)}}
\newcommand{\figfref}[1]{\hyperref[#1]{Fig.~\ref{#1}(f)}}
\renewcommand{\eqref}[1]{\hyperref[#1]{Eq.~(\ref{#1})}}
\newcommand{\secref}[1]{\hyperref[#1]{Sec.~\ref{#1}}}
\newcommand{\eqsref}[2]{\hyperref[#1]{Eqs.~(\ref{#1})-(\ref{#2})}}
\newcommand{\appref}[1]{\hyperref[#1]{Appx.~\ref{#1}}}

\begin{document}

\title{Stabilizer codes for Heisenberg-limited many-body Hamiltonian estimation}

\author{Santanu Bosu Antu}
\affiliation{Perimeter Institute for Theoretical Physics, Waterloo, Ontario N2L 2Y5, Canada}
\affiliation{Department of Applied Physics and Yale Quantum Institute, New Haven, Connecticut 06520, USA}

\author{Sisi Zhou}
\email{sisi.zhou26@gmail.com}
\affiliation{Perimeter Institute for Theoretical Physics, Waterloo, Ontario N2L 2Y5, Canada}
\affiliation{Department of Physics and Astronomy, Department of Applied Mathematics, and Institute for Quantum Computing, University of Waterloo, Ontario N2L 3G1, Canada}


\begin{abstract}

Estimating many-body Hamiltonians has wide applications in quantum technology. By allowing coherent evolution of quantum systems and entanglement across multiple probes, the precision of estimating a fully connected $k$-body interaction can scale up to $(n^kt)^{-1}$, where $n$ is the number of probes and $t$ is the probing time. However, the optimal scaling may no longer be achievable under quantum noise, and it is important to apply quantum error correction in order to recover this limit. In this work, we study the performance of stabilizer quantum error correcting codes in estimating many-body Hamiltonians under noise. When estimating a fully connected $ZZZ$ interaction under single-qubit noise, we showcase three families of stabilizer codes---thin surface codes, quantum Reed--Muller codes and Shor codes---that achieve the scalings of $(nt)^{-1}$, $(n^2t)^{-1}$ and $(n^3t)^{-1}$, respectively, all of which are optimal with $t$. We further discuss the relation between stabilizer structure and the scaling with $n$, and identify several no-go theorems. For instance, we find codes with constant-weight stabilizer generators can at most achieve the $n^{-1}$ scaling, while the optimal $n^{-3}$ scaling is achievable if and only if the code bears a repetition code substructure, like in Shor code. 
\end{abstract}

\maketitle


\section{Introduction}

High-precision measurement of parameters, or metrology, is crucial in all branches of science, especially in physics~\cite{giovannetti2011advances,degen2017quantum,pezze2018quantum,pirandola2018advances,jiao2023quantum}. Classically, a common method to estimate an unknown parameter is to repeat the same experiment multiple times. As a consequence of the central limit theorem, the estimation error will scale as $1/\sqrt{t}$, known as the standard quantum limit (SQL), where $t$ is the time of experiment. Quantum mechanics allows us to achieve better precision by adopting the advantage of long-time coherent evolution, which results in an estimation error of order $1/t$, known as the \emph{Heisenberg limit} (HL)~\cite{giovannetti2004quantum,giovannetti2006quantum} which quadratically outperforms the SQL. 

The HL, however, can be fragile against quantum noise~\cite{huelga1997improvement,escher2011general,demkowicz2012elusive,demkowicz2017adaptive,zhou2018achieving}. To address this issue, quantum error correction (QEC) has been proposed as a tool to recover the HL under noise~\cite{kessler2014quantum,arrad2014increasing,dur2014improved,sekatski2017quantum,demkowicz2017adaptive,zhou2018achieving,layden2019ancilla,zhou2020optimal,zhou2021asymptotic,zhou2024achieving}. Specifically, to estimate a Hamiltonian parameter under Markovian noise, it was known the HL is achievable if and only if the ``Hamiltonian-not-in-Lindblad-span'' (HNLS) condition is satisfied~\cite{demkowicz2017adaptive,zhou2018achieving}, in which case a QEC protocol was proposed to recover the HL. Despite the generality of the QEC protocol, its implementation can be challenging in practice, partially due to the lack of efficient encoding and decoding procedure, and the requirement of a noiseless ancillary system as part of the code~\cite{zhou2018achieving,zhou2021asymptotic}. Using stabilizer codes~\cite{gottesman1997stabilizer,nielsen2001quantum}, a class of QEC codes with natural encoding and decoding methods via stabilizer measurements and no requirement of noiseless ancilla, for quantum metrology can alleviate the above constraints. So far, stabilizer codes have rarely been explored for quantum metrology in previous works~\cite{zhou2018achieving,layden2019ancilla,zhou2020optimal,zhou2021asymptotic,zhou2024achieving}, except for repetition codes for phase estimation against bit-flip noise~\cite{kessler2014quantum,arrad2014increasing,dur2014improved}. 

In multi-probe systems, estimating the strength of Hamiltonians has always been a major topic in quantum metrology~\cite{wineland1992spin,bollinger1996optimal,leibfried2004toward,higgins2007entanglement,giovannetti2004quantum,giovannetti2006quantum,kaubruegger2021quantum,huang2023learning,dutkiewicz2023advantage,boixo2007generalized}. The scaling of the estimation error with respect to the number of probes $n$ can achieve up to $1/n$ for 1-local Hamiltonians, using entangled states, e.g., the Greenberger--Horne--Zeilinger (GHZ) state~\cite{giovannetti2006quantum} and we will refer to it as the \emph{Heisenberg scaling} (HS)\footnote{Note that the scaling $1/n$ can also sometimes be referred to as the Heisenberg limit (with respect to the number of probes). However, in this paper we call it the HS to distinguish it from the HL with respect to probing time $t$.}.  

Interestingly, the estimation with higher-order interactions can go beyond the HS~\cite{boixo2007generalized,beau2017nonlinear,czajkowski2019many} (even without entanglement~\cite{boixo2008quantum}). For example, estimating a $k$-local Hamiltonian can reach a precision of order $1/n^k$, usually referred to as the \emph{super-Heisenberg scaling} (SHS) when $k > 1$~\cite{boixo2007generalized}. Similar to the HL, the HS and SHS are also compromised by the effect of quantum noise. For example, the HS for estimating 1-local Hamiltonians will degrade to $1/\sqrt{n}$ under generic single-qubit noise~\cite{demkowicz2012elusive,escher2011general,demkowicz2017adaptive,zhou2018achieving}. In general, it is of great interest to explore which types of SHS are recoverable for estimating many-body Hamiltonians under different types of quantum noise.

In this work, we consider the achievability of the HL, HS and SHS for estimating a fully connected $ZZZ$ interaction under single-qubit noise using stabilizer codes. First, we will present three classes of $[[n,1,3]]$ stabilizer codes---thin surface codes~\cite{kitaev2003fault,dennis2002topological,fowler2012surface}, quantum Reed-Muller codes~\cite{steane1999quantum} and Shor codes~\cite{shor1995scheme}---that achieve the HL and HS. This demonstrates the applicability of stabilizer codes in noisy quantum metrology. In particular, the optimal SHS $1/n^3$ can be achieved by Shor codes. Secondly, we will dive into the achievability of the SHS and explore the roles of different code structures, e.g., non-locality of stabilizer generators, repetition code substructures, etc., that are necessary for achieving different types of SHS. One prominent result is the non-achievability of the SHS when stabilizer generators have constant weights, e.g., for quantum low-density parity-check (LDPC) codes~\cite{breuckmann2021quantum}. These no-go results then imply the optimality of the previously mentioned three classes of codes under different constraints on code structures, and may further provide guidance on identifying or constructing stabilizer codes that are suitable for many-body Hamiltonian estimation.

Our work studies Hamiltonian estimation under Markovian noise during time evolution with the novel stabilizer QEC technique. Here we comment on the relationship between our work and previous literature. First, this work focuses on Hamiltonian estimation from time evolution of the Hamiltonian~\cite{huang2023learning,dutkiewicz2023advantage,boixo2007generalized,granade2012robust,hincks2018hamiltonian,wiebe2014hamiltonian,evans2019scalable,yu2023robust,li2020hamiltonian,hangleiter2024robustly,stilck2024efficient}, where input states, quantum controls during the evolution and measurements can be chosen by the observer, a framework different from Hamiltonian estimation from Gibbs states~\cite{anshu2020sample,haah2022optimal,bairey2019learning,qi2019determining}. In our case, the input states, quantum controls, and measurements are logical states, quantum error correction, and logical measurements, respectively, and our goal is the achieve the HL with respect to the probing time~\cite{huang2023learning,dutkiewicz2023advantage,boixo2007generalized}. Second, the noise model we consider is the Markovian noise during the evolution of the Hamiltonian, which is different from other literature that primarily focused on the state preparation, measurement, and gate implementation errors~\cite{wiebe2014hamiltonian,evans2019scalable,yu2023robust}. Our results represent the ultimate achievable estimation limits under stochastic noise when perfect quantum controls are available. Third, we treated the Hamiltonian estimation problem as a single-parameter estimation problem~\cite{boixo2007generalized} where the strength of an all-to-all $ZZZ$ interaction is to be measured and we analyze the scaling of the estimation precision with respect to the number of qubits. This is in contrast to (although related to) the multi-parameter approach~\cite{huang2023learning,dutkiewicz2023advantage,granade2012robust,hincks2018hamiltonian,wiebe2014hamiltonian,evans2019scalable,yu2023robust,li2020hamiltonian,hangleiter2024robustly,stilck2024efficient} where the goal is to show identifiability of the Hamiltonian coefficients, and to optimize the sample complexity with respect to the system size and the estimation error. Finally, our work is based on the point estimation approach using the quantum Cram\'{e}r--Rao bound, where the range of the unknown parameter is assumed to be priorly narrowed down to a small neighbourhood and our goal is to perform high-precision estimation in the neighbourhood. The Cram\'{e}r–Rao approach characterizes the asymptotic behaviour of parameter estimation when the number of experiments is sufficiently large. This contrasts with the Bayesian approach~\cite{jarzyna2015true,rubio2018non}, which selects estimators, measurement schemes, and initial states adaptively based on prior probabilities~\cite{granade2012robust, wiebe2014hamiltonian,evans2019scalable,hincks2018hamiltonian} and is more relevant in scenarios where the target estimation precision is a constant and the number of allowed experiments is limited.

\section{Formulation of the Problem}\label{sec: formulation of the problem}

Hamiltonian evolution under Markovian noise can be described by the Lindblad master equation~\cite{gorini1976completely,lindblad1976generators,breuer2002theory}
\begin{equation}
\label{eq:master}
\frac{d\rho}{dt} = -i[H,\rho] + \sum_{a} \left(L_a \rho L_a^\dagger - \frac{1}{2}\{L_a^\dagger L_a,\rho\}\right), 
\end{equation}
where $\rho$ is the density matrix of the probe system, $H$ is the Hamiltonian and $\{L_a\}_a$ are the Lindblad operators describing quantum noise. We consider the situation where the probe system consists of $n$ qubits ($n\geq3$), 
\begin{equation}
\label{eq:hamt}
    H = \omega \sum_{1 \leq i<j<k\leq n} Z_iZ_jZ_k + H_0 \coloneqq \omega G + H_0,
\end{equation}
where $H_0$ is any (unknown) 2-local Hamiltonian, $\omega$ is the strength of the 3-local Hamiltonian $G$ that we want to estimate, and $Z_i$ represents the Pauli-Z operator acting on the $i$-th qubit (we will use $X,Y,Z$ to denote Pauli-X,Y,Z operators, respectively). The set $\{L_a\}_a$ consists of local operators (i.e., each $L_a$ acts non-trivially on only one qubit). We further assume that the noise is complicated enough that the Lindblad span 
\begin{align}
    \mS &\coloneq {\rm span}_{\rm H}\{I,L_a,L_{a'}^\dagger,L_a^\dagger L_{a'}\}_{\forall {a,a'}} \\ & = {\rm span}_{\rm H}\{P_iP_j\}_{\forall {i,j;P\in\{X,Y,Z,I\}}}
\end{align}contains all two-local Hermitian operators (${\rm span}_{\rm H}\{\cdot\}$ denotes the set of all Hermitian operators that can be written as linear combinations of operators in $\{\cdot\}$). It is a natural assumption of quantum noise, and is satisfied by, for example, single-qubit depolarizing noise, single-qubit amplitude damping noise, etc. 

In general, the HNLS condition~\cite{zhou2018achieving} states that it is possible to achieve the HL if and only if $G \not\in \mS$, which is satisfied in our setting. Due to the HNLS condition~\cite{zhou2018achieving}, our scenario represents one of the \emph{simplest} cases of many-body Hamiltonian estimation where a QEC code can be used to recover the HL under single-qubit noise. This is because, for arbitrary single-qubit errors, the Lindblad span includes all two-body Pauli terms, and, consequently, the HL cannot be recovered for any two-body Hamiltonian estimation. This motivates our focus on three-body Hamiltonians, which lie just outside of the Lindblad span and therefore allow for Heisenberg-limited estimation using QEC.

Consider the quantum metrological protocol where an initial $n$-qubit state $\rho(0)$ is prepared, quantum controls are applied sufficiently fast during the evolution, and $\omega$ is estimated through quantum measurement on the output state $\rho_\omega(t)$ at time $t$~\cite{sekatski2017quantum,demkowicz2017adaptive,zhou2018achieving}. The estimation precision $\delta\omega$ can be given by the quantum Cram\'{e}r--Rao bound~\cite{holevo2011probabilistic,helstrom1969quantum} 
\begin{equation}
    \delta\omega \geq 1/\sqrt{\nu F(\rho_\omega(t))},
\end{equation} 
where $\nu$ is the number of repeated experiments and $F(\rho_\omega(t))$ is the quantum Fisher information (QFI) of the state $\rho_\omega(t)$~\cite{braunstein1994statistical,paris2009quantum}. The quantum Cram\'{e}r--Rao bound is saturable as $\nu \rightarrow \infty$~\cite{paris2004quantum}, and thus the scaling of $F(\rho_\omega(t))$ determines the scaling of $\delta\omega$ with respect to $t$ and $n$. 
Specifically, we say that the HL is achieved when $\delta \omega \sim t^{-1}$ or $F(\rho_\omega(t)) = \Theta(t^2)$, the HS is achieved when $\delta \omega \sim n^{-1}$ or $F(\rho_\omega(t)) = \Theta(n^2)$, and the SHS is achieved when $\delta \omega \sim n^{-k}$ or $F(\rho_\omega(t)) = \Theta(n^{2k})$ for any $k > 1$. 

In the noiseless case ($L_a = H_0 = 0$), the optimal initial state is the GHZ state 
\begin{equation}
\ket{\psi(0)} = \frac{1}{
\sqrt{2}}({\ket{0^{\otimes n}}+\ket{1^{\otimes n}}}),
\end{equation}
which leads to an output state 
\begin{equation}
\ket{\psi_\omega(t)} = \frac{e^{i\omega \frac{n(n-1)(n-2)}{6}}\ket{0^{\otimes n}}+e^{-i\omega \frac{n(n-1)(n-2)}{6}}\ket{1^{\otimes n}}}{
\sqrt{2}},
\end{equation}
after evolution time $t$, and the corresponding QFI of the output state $\rho_\omega(t) = \ket{\psi_\omega(t)}\bra{\psi_\omega(t)}$ is 
\begin{equation}
    F(\rho_\omega(t)) = 4 \left(\frac{n(n-1)(n-2)}{6}\right)^2 t^2 = \Theta(n^6 t^2),
\end{equation}
which achieves the HL and the SHS.

With single-qubit errors, the QEC protocol utilizes a QEC code that perfectly corrects them, i.e., $\Pi S \Pi \propto \Pi$ for all $S \in \mS$, where $\Pi$ is the projection onto the code subspace~\cite{bennett1996mixed,knill1997theory}. In the fast control limit (i.e., fast quantum controls can be applied at an arbitrarily high frequency), the evolution of $\rho_\tl=\Pi \rho \Pi$ is given by
\begin{equation}
    \frac{d\rho_\tl}{dt} = -i\omega [G_\eff,\rho_\tl],
\end{equation}
where $G_\eff = \Pi G\Pi \not\propto \Pi$. (Note that in the noiseless case, $\Pi = I$ and the discussion below also holds.)
The QFI is maximized when the initial state is chosen as $\ket{\psi(0)} = \frac{1}{
\sqrt{2}}\left(\ket{\lambda_{\min}}+\ket{\lambda_{\max}}\right)$, 
that leads to
\begin{equation}
    F(\rho_\omega(t)) = (\lambda_{\max}-\lambda_{\min})^2 t^2 \eqcolon (\Delta G_\eff)^2 t^2, 
    \label{eqn: Fisherlambda_minmax}
\end{equation}
which achieves the HL, where $\lambda_{{\min},{\max}}$ are the smallest (and largest) eigenvalues that correspond to the eigenstates $\ket{\lambda_{{\min},{\max}}}$ of $G_\eff$ that are restricted in the code subspace. $\Delta G_\eff := \lambda_{\max}-\lambda_{\min}$. Note that the above implies that it is sufficient to consider a two-level logical system, as we can always, without loss of generality, restrict the code subspace to ${\rm span}\{\ket{\lambda_{{\min}}},\ket{\lambda_{{\max}}}\}$. 

Among all possible ancilla-assisted QEC codes,
the optimal QEC protocol~\cite{zhou2018achieving} can achieve $(\Delta G_\eff)_{\opt} = 2 \norm{G-\mS}$, where $\norm{G-\mS}$ is the distance between $G$ and $\mS$ in terms of the operator norm. Specifically in our case (see \appref{app:opt-qfi}), 
\begin{equation}
\label{eq:opt-coeff}
    (\Delta G_\eff)_{\opt} = 2\!\min_{\text{2-local }\!S} \norm{G - S} = \frac{n^3}{12} + O(n^2). 
\end{equation}
Comparing this to the noiseless case, where 
\begin{equation}
    \Delta G = \frac{n^3}{3} + O(n^2), 
\end{equation}
we see that the optimal QFI for estimating $ZZZ$ interaction under generic qubit noise is $F(\rho_\theta(t)) = \Theta(n^6t^2)$, which bears the same scaling as the noiseless one.

Although QEC can in principle recover the HL and the SHS for estimating the 3-local Hamiltonian $G$ from the above discussion, it is unclear how the QEC protocol (with encoding, decoding and measurement operations) can be implemented in reality. 
In this work, we consider stabilizer QEC for the task of estimating $ZZZ$ interactions, which is the most popular family of QEC and was demonstrated recently in experiments across multiple platforms~\cite{google2023suppressing,ryan2021realization,bluvstein2024logical}. One advantage of it is the standard encoding, decoding and measurement procedures via syndrome measurements~\cite{gottesman1997stabilizer} that were currently unavailable in the known QEC protocols for metrology. In addition, as stabilizer codes have already been theoretically extensively studied in previous literature for quantum memory, quantum computation, etc., it is natural to explore their applications in other noisy information processing tasks like quantum metrology and identify the code properties that are crucial for many-body Hamiltonian estimation.

\section{Achieving the HL, HS and SHS}
\label{sec:achieving}
As discussed in the previous section, the HL is achievable using a QEC code  if and only if $G_\eff = \Pi G \Pi$ acts non-trivially in the code subspace. In this section, we first establish a relation between $\Delta G_\eff$ and the logical operators of the stabilizer code, defined by an abelian subgroup $\scrS$ of the $n$-qubit Pauli group $\{\pm 1,\pm i\} \times \{I,X,Y,Z\}^{\otimes n}$ which stabilizes the code space, and then present three families of stabilizer codes that achieve both the HL and the HS (or SHS). 

\subsection{Relation between QFI and \texorpdfstring{$Z^{\otimes 3}$}{Z3}-type Logical Operators}

Here we use the stabilizer formalism to establish a relation between the QFI and the number of \emph{$Z^{\otimes 3}$-type logical operators}, i.e., operator $Z_iZ_jZ_k$ for some $i,j,k$ that maps logical states to logical states in a non-trivial way, and this relation will be used throughout the remainder of the paper.

\begin{lemma}
Consider using an $[[n,1,3]]$\footnote{We use $[[n,k]]$ ($[[n,k,d]]$) to represent a stabilizer code that encodes $k$ logical qubits into $n$ physical qubits (with distance $d$).} stabilizer code to estimate $\omega$ in \eqref{eq:master} and \eqref{eq:hamt}. Let $\ell$ be the number of $Z_i Z_j Z_k$ operators appearing in $G$ that are logical operators of the code. Then the output QFI 
\begin{equation}
    F(\rho_\omega(t)) \leq 4\ell^2t^2. 
\end{equation}
The equality holds when the stabilizer code satisfies
\begin{property}
\label{prop}
Any $Z$-type stabilizer (i.e., stabilizer that consists of only $I$ and $Z$) has a positive sign\footnote{For example, $ZZIII$ has a positive sign; $-ZZIII$ has a negative sign. ($\pm i ZZIII$ cannot be in any stabilizer group because their squares are $-I$.)}. 
\end{property}

\label{lemma:QFIlogicalrelation}
\end{lemma}
\begin{proof}
    Let $\scrS$ and $\scrN(\scrS)$ be the stabilizer and normalizer sets of the stabilizer code, respectively. A term $Z_iZ_jZ_k$ of $G$ belongs to one of the following three classes: (1)~$Z_iZ_jZ_k\in \scrS$, (2)~$Z_iZ_jZ_k\not\in \scrN(\scrS)$, and (3)~$Z_iZ_jZ_k\in$ $\scrN(\scrS)\setminus\scrS$. The operator is called a \emph{logical operator} if it belongs to class~(3). By \eqref{eqn: Fisherlambda_minmax}, the QFI is proportional to $\Delta G_\eff$ where $G_\eff = \Pi G \Pi = \sum_{ijk} \Pi Z_iZ_jZ_k \Pi$ with $\Pi$ being the projection onto the code space. We separately calculate the contributions of the three types of terms to the $\Delta G_\eff$.
    
    \begin{enumerate}[wide, labelwidth=!,itemindent=!,labelindent=0pt, leftmargin=0em, label={(\arabic*)}, parsep=0pt]
        \item Suppose $Z_iZ_jZ_k\in \scrS$. It immediately follows that $Z_iZ_jZ_k \Pi = \Pi$ such that $\Pi Z_iZ_jZ_k \Pi = \Pi$. Hence, the terms that are in $\scrS$ do not contribute to $\Delta G_\eff$.
        
        \item Suppose that $Z_i Z_j Z_k \not \in \scrN(\scrS)$. Then, there is an operator $Q\in \scrS$ that anti-commutes with $Z_iZ_jZ_k$. 
        So, $\Pi Z_iZ_jZ_k \Pi = \Pi Q Z_iZ_jZ_k \Pi = - \Pi Z_iZ_jZ_k Q \Pi = - \Pi Z_iZ_jZ_k \Pi$, which leads to $\Pi Z_iZ_jZ_j \Pi = 0$. Hence, the terms that are not in the normalizer group also do not contribute to the QFI.

        \item Suppose that $Z_iZ_jZ_k\in \scrN(\scrS)\setminus \scrS$. As all the $ZZZ$ terms commute with each other, the non-trivial logical operators must be different representatives of the same logical operator up to a $\pm 1$ coefficient. In addition, $\Pi Z_iZ_jZ_k \Pi$ must correspond to a logical Pauli operator in the two-dimensional code subspace, which means $\Delta ( \Pi Z_iZ_jZ_k \Pi) = 2$.
    \end{enumerate}
    Hence, if the number of $Z^{\otimes 3}$-type logical operators in $G$ is $\ell$, the QFI is given by $F(\rho_\omega(t)) = (\Delta G_\eff)^2 t^2 \leq 4\ell^2t^2$. When \propref{prop} holds, all $Z^{\otimes 3}$-type logical operators must correspond to the same logical operator with the same sign, and the equality holds. 
\end{proof}

We note that any stabilizer code can be modified to satisfy \propref{prop} by only changing the signs of its stabilizer generators. Thus, we always assume, without the loss of generality, that \propref{prop} holds in all discussions below. For example, the surface code, the Reed--Muller code and the Shor code we consider below all naturally satisfy \propref{prop}. Given a stabilizer code, to achieve \propref{prop}, we first find a set of its stabilizer generators which has the maximal number of $Z$-type stabilizers. That means the rest of the stabilizer generators in the set cannot generate a $Z$-type stabilizer. This set can be efficiently obtained by Gaussian elimination. Then we modify all signs of these $Z$-type stabilizer generators to the positive sign. For example, if we have a stabilizer generator $-ZZIII$, we switch the sign to $+1$ and define a new stabilizer code where $-ZZIII$ is replaced by $ZZIII$. This guarantees \propref{prop} holds because any $Z$-type stabilizer can only be expressed as products of $Z$-type stabilizer generators and thus also must has a positive sign. 

\lemmaref{lemma:QFIlogicalrelation} implies that the HL is achievable for any stabilizer codes with $Z^{\otimes 3}$-type logical operators. The scaling of the QFI with respect to $n$ is determined by the scaling of $\ell$ with respect to $n$. In the following parts of this section, we will apply \lemmaref{lemma:QFIlogicalrelation} to demonstrate the achievability of the HS ($F(\rho_\omega(t)) = \Theta(n^2)$) and SHS ($F(\rho_\omega(t)) = \Theta(n^4),\Theta(n^6)$) using three different families of stabilizer codes, or more specifically, Calderbank--Shor--Steane (CSS) codes~\cite{calderbank1996good,steane1996error}. 

Note that \lemmaref{lemma:QFIlogicalrelation} automatically applies to the estimation of higher-order $Z^{\otimes k}$ terms, i.e. when $G = \sum_{\{i_k\}\subseteq\{1,2,\ldots,n\}} Z_{i_1}Z_{i_2}\cdots Z_{i_k}$ with $k \geq 3$, as the proof of \lemmaref{lemma:QFIlogicalrelation} is still correct when we replace $Z_iZ_jZ_k$ with higher-order $Z$-type operators $Z_{i_1}Z_{i_2}\cdots Z_{i_k}$.

\subsection{Thin Surface Code}\label{sec:thin-surface}

First, we consider a special kind of surface code~\cite{kitaev2003fault,dennis2002topological,fowler2012surface} that we will refer to as the \textit{thin surface code}, which corrects single-qubit $Z$ errors. \figref{fig:planarstabs} shows the $[[n,1,3]]$ thin surface code with $n = 33$. The qubits are placed on the edges of a $(2 \times (n-3)/5)$ lattice. On the primary (solid) lattice, the $X$-type stabilizer generators are $X^{\otimes 4}$ plaquette operators in the bulk and on the smooth edges, and $X^{\otimes 3}$ plaquette operators on the rough edges. Similarly, the $Z$-type stabilizer generators are $Z^{\otimes 4}$ and $Z^{\otimes 3}$ plaquette operators on the dual (dashed) lattice, 

\begin{figure}[tbp]
\centering
\includegraphics[width=0.48\textwidth]{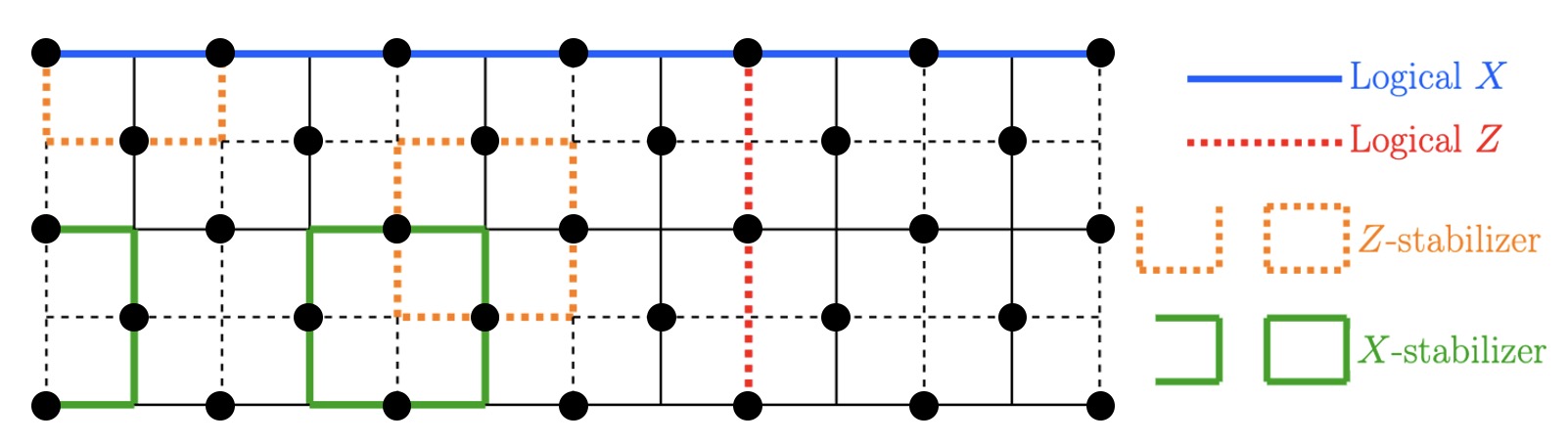}
\caption{$[[33,1,3]]$ thin surface code. Physical qubits are placed on the edges of the solid lattice (including the start and end points of the horizontal solid lines). $X$/$Z$-type stabilizer and logical operators are represented by solid/dashed lines that pass through the involved qubits.}
\label{fig:planarstabs}
\end{figure}

The quantum code encodes one logical qubit into $n$ physical qubits. A logical $X$ operator is a string of $X$ operators connecting the two rough edges on the primary lattice. Similarly, a logical $Z$ operator is a string of $Z$ operators connecting the two smooth edges on the dual lattice.

For the $[[n,1,3]]$ thin surface code, the $Z^{\otimes 3}$-type logical operators must correspond to one vertical $Z^{\otimes 3}$ strings on the dual lattice. This is because multiplying any one of them by a stabilizer can either increase its weight, or transform it to another $Z^{\otimes 3}$-type logical operator. By a simple counting, the number of such $Z^{\otimes 3}$ strings for an $n$-qubit code is $(n+2)/5$. Therefore, by \lemmaref{lemma:QFIlogicalrelation}, we have the following result. 
\begin{theorem}
    The $[[n,1,3]]$ thin surface code achieves a QFI of ${4(n+2)^2t^2}/{25} = \Theta(n^2t^2)$.
    \label{thm: thin planar scaling}
\end{theorem}

We also note that the result above can readily generalize to estimating $k$-body terms in Hamiltonian (when $k$ is a constant independent of $n$ and $n$ is sufficiently large). We observe that the $[[n,1,3]]$ thin surface code has $\Theta(n^{\lfloor k/3 \rfloor})$ $Z^{\otimes k}$-type logical operators, implying a QFI of $\Theta(n^{2\lfloor k/3 \rfloor}t^2)$. To prove this, we can divide the entire $n$ qubits into $\Theta(n)$ constant-size groups, e.g. groups of qubits in a $5 \times 5$ grid, with each group sufficiently large to accommodate at least one $Z^{\otimes 3}$-type logical operator, $Z^{\otimes 3}$-type, $Z^{\otimes 4}$-type and $Z^{\otimes 5}$-type stabilizers. On one hand, we note that $\Omega(n^{\lfloor k/3 \rfloor})$ different $Z^{\otimes k}$-type logical operators can be constructed by picking $\lfloor k/3 \rfloor$ groups out of the entire $\Theta(n)$ groups and choose one $Z^{\otimes 3}$-type logical operator, $\lfloor k/3 \rfloor - 2$ $Z^{\otimes 3}$-type stabilizers and one $Z^{\otimes k - 3\lfloor k/3 \rfloor + 3}$-type stabilizer from each group. On the other hand, the number of $Z^{\otimes k}$-type logical operators is at most $O(n^{\lfloor k/3 \rfloor})$ because any $Z^{\otimes k}$-type logical operator can be supported on at most $\lfloor k/3 \rfloor$ groups, and the number of subsets containing at most $\lfloor k/3 \rfloor$ groups is $\Theta(n^{\lfloor k/3 \rfloor})$.

\subsection{Quantum Reed-Muller Code}\label{subsec:QRM}

The stabilizer group of a CSS code is generated by operators from two subsets $\sfS_x$ and $\sfS_z$, where the subset $\sfS_x$ ($\sfS_z$) consists of stabilizer generators that are tensor products of $X$ ($Z$) and $I$.
Here, we consider \textit{quantum Reed--Muller (QRM) codes} that are CSS codes whose $\sfS_x$ and $\sfS_z$ are obtained from the parity-check matrices of two \textit{classical Reed--Muller (RM) codes}~\cite{steane1999quantum, anderson2014fault}.

We first review the definition of classical RM codes~\cite{macwilliams1977theory}. Define $\mbv_0^T=(1,1,\ldots,1)\eqcolon \bm{1}_{2^m}^T$ as the $2^m$-tuple row vector consisting entirely of 1s. Throughout this section, we will use the notation $\bm{0}_{p}$ to denote a vector of $p$ zeros, and $\bm{1}_{p}$ to denote a vector of $p$ ones. For any integer $i$ with $1\leq i\leq m$, define the $2^m$-tuple row vector $\mbv_i^T=(\bm{0}_{2^{i-1}}^T|\bm{1}_{2^{i-1}}^T| \bm{0}_{2^{i-1}}^T| \ldots| \bm{1}_{2^{i-1}}^T)$, where ``$\,|\,$'' denotes the concatenation of row vectors. The generator matrix of $RM(1,m)$, i.e., the 1st order RM code of length $2^m$, is given by
\begin{align}
    \sfG(1,m)= \begin{pmatrix}
  \text{---} \bm{v}_0^T \text{---}\\
  \text{---} \bm{v}_1^T \text{---}\\
  \vdots\\
  \text{---} \bm{v}_m^T \text{---}\\
    \end{pmatrix}.
\end{align}
The $r$-th order RM code, denoted by $RM(r,m)$, is constructed by taking the Boolean products (denoted by $*$) of the codewords of a 1st-order RM code. In particular, the generator matrix $\sfG(r,m)$ of $RM(r,m)$ is obtained by taking the rows of $RM(1,m)$ along with all the distinct $2,3, \ldots, r$-fold Boolean products of the rows of $\sfG(1,m)$. For example, the set of rows of the generator matrix $\sfG(2,3)$ is $\{\bm{v}_0^T, \bm{v}_1^T, \bm{v}_2^T, \bm{v}_3^T, \bm{v}_1^T*\bm{v}_2^T, \bm{v}_1^T*\bm{v}_3^T, \bm{v}_2^T*\bm{v}_3^T\}$, where $\sfG(1,3)$ corresponds to the first four rows. It can be shown that $RM(r,m)$ is a $[2^m,N(m,r),2^{m-r}]$ code, where $N(m,r)\coloneqq 1+ \binom{m}{1}+\binom{m}{2}+ \cdots + \binom{m}{r}$ and $RM(r,m)^\perp =RM(m-r-1,m)$~\cite{macwilliams1977theory}.  

The QRM codes are derived from shortened RM codes $\overline{RM}(r,m)$ whose  generating matrix $\overline{\sfG}(r,m)$ is obtained by deleting the first row and the first column of $\sfG(r,m)$. The following three properties of a shortened RM code $\overline{RM}(r,m)$ will be useful to us. Their proofs are given in \appref{app:A}.

\begin{enumerate}[wide, labelwidth=!,itemindent=!,labelindent=0pt, leftmargin=0em, label={(\arabic*)}, parsep=0pt]
    \item $\overline{RM}(1,m)$ is strictly contained in $\overline{RM}(m-2,m)$. 
    \item The vector $\bm{1}_{2^m-1}$ and those corresponding to the rows of $\overline{\sfG}(m-2,m)$ form a basis for $\overline{RM}(1,m)^\perp$. 
    \item The number of codewords in $\overline{RM}(1,m)^\perp$ with weight $(2^m-4)$ is $\frac{1}{6}(4^m-3\cdot 2^m+2)$.
\end{enumerate}
Here, we focus on $QRM(1,m)$ which is an $[[n=2^m-1,1,3]]$ code (see a definition of general $QRM(r,m)$ in e.g., \cite{zeng2011transversality}). The set of $X$-type stabilizer generator $\sfS_x$ is obtained from the rows of $\overline{\sfG}(1,m)$ by substituting $1$s by $X$s and $0$s by $I$s in the tensor product forms of the operators. The elements of $\sfS_z$ are similarly obtained from $\overline{\sfG}(m-2,m)$ by replacing $1$s by $Z$s and $0$s by $I$s. Due to Property~(1), the basis codewords of the CSS code $QRM(1,m)$ are given by $\ket{z+\overline{RM}(1,m)}=\frac{1}{2^m}\sum_{y\in \overline{RM}(1,m)}\ket{z+y}$ for $z\in \overline{RM}(m-2,m)$~\cite{nielsen2001quantum}. Note that $Z^{\otimes n}$ commutes with all the $X$-type stabilizer generators since each generator is of even weight; however, $Z^{\otimes n}$ is not in the stabilizer group by Property~(2). 
Thus, $Z^{\otimes n}$ is a logical operator and one can choose a specific basis in the logical subspace such that it is a logical $Z$ operator, i.e., 
\begin{align}
    \overline{Z}:=Z^{\otimes n}=Z^{\otimes (2^m-1)}.
    \label{eqn: RMlogicalZ}
\end{align}
Following \lemmaref{lemma:QFIlogicalrelation}, we obtain the scaling of the QFI with respect to $n$ by calculating the number of weight-3 representatives of the logical Pauli-Z operator.

\begin{theorem}
The $[[n = 2^m - 1,1,3]]$ $QRM(1,m)$ code achieves a QFI of ${n^2(n-1)^2}t^2/{9} = \Theta(n^4t^2)$. 
\label{thm: qrm}
\end{theorem}

\begin{proof}
    Operators equivalent to the logical $Z$ operator in \eqref{eqn: RMlogicalZ} are obtained by multiplying $\overline{Z}$ with different stabilizers. Since the weight of $\overline{Z}$ is $2^m-1$, it follows that the number of $Z^{\otimes 3}$-type logical operators $\ell$ is the same as the number of $Z^{\otimes(n-3)} = Z^{\otimes(2^m - 4)}$ stabilizers. This amounts to calculating the number of weight-$(2^m-4)$ codewords in the shortened code $\overline{RM}(m-2,m)$.

    Below we prove the number of weight-$(2^m-4)$ codewords in  $\overline{RM}(m-2,m)$ is equal to that in $\overline{RM}(1,m)^\perp$. Using Property~(3), this implies that $\ell = \frac{1}{6}(4^m-3\cdot 2^m+2) = \frac{1}{6}n(n-1)$. Combining with \lemmaref{lemma:QFIlogicalrelation}, the theorem is then proven. 

    To prove the statement above, let $V=\{\bm{v}_1, \bm{v}_2,\ldots, \bm{v}_q\}$ be the set of weight-$(2^m-4)$ vectors in $\overline{RM}(m-2,m)$. Using Property~(2), it follows that $V \cup \{\bm{1}_n+\bm{u}_i: 1\leq i \leq p \}$ is the set of weight-$(2^m-4)$ vectors in $\overline{RM}(1,m)^\perp$, where $U=\{\bm{u}_1, \bm{u}_2,\ldots, \bm{u}_p\}$ is the set of weight-3 vectors in $\overline{RM}(m-2,m)$. We know that the weights of the codewords of the unshortened $RM(m-2,m)$ are even\footnote{Can be shown using the Macwilliams identity~\cite{macwilliams1977theory} (see \appref{app:A}).}. Therefore, the weight-3 codewords of $\overline{RM}(m-2,m)$ are exactly the shortened versions of the weight-4 codewords of $RM(m-2,m)$ that have 1 as the first entry. So, for each $\bm{u}_i\in U$, we have a distinct weight-4 vector $(1|\bm{u}_i^T)^T\in RM(m-2,m)$. This vector is of the form $(1|\bm{u}_i^T)^T=\bm{1}_{2^m}+\bm{w}$, where $\bm{w}\in RM(m-2,m)$ is a weight-$(2^m-4)$ vector whose first entry is 0. Hence, $\bm{w}^T=(0|\bm{v}_j^T)$ for some $\bm{v}_j\in V$. This implies that $\bm{v}_j=\bm{1}_n+\bm{u}_i$. Therefore, $V \cup \{\bm{1}_n+\bm{u}_i: 1\leq i \leq p \}=V$, thus completing the proof. 
\end{proof}
 
The QRM code achieves the SHS, which outperforms the thin surface code. As we will see, we can achieve a better scaling using Shor codes. 

\subsection{Shor Code and Generalization}

We now consider the $n$-qubit Shor code, which is a simple generalization of the $9$-qubit Shor code~\cite{shor1995scheme}, where $n = 3 n_r$ is divisible by $3$. The code is constructed by concatenating a phase flip code $\left(\frac{\ket{0}\pm\ket{1}}{\sqrt{2}}\right)^{\otimes 3}$ of distance three with an $n_r$-qubit repetition code. (Note that here a concatenation of code 1 and code 2 means encoding every physical qubit of code 1 using code 2.) The logical basis states are given by 
\begin{subequations}
\label{eqn: Shorlogbasis}
\begin{align}
    \ket{0_\tl}=\frac{1}{\sqrt{8}}\left(\ket{0}^{\otimes n_r}+\ket{1}^{\otimes n_r}\right)^{\otimes 3},\\
    \ket{1_\tl}=\frac{1}{\sqrt{8}}\left(\ket{0}^{\otimes n_r}-\ket{1}^{\otimes n_r}\right)^{\otimes 3}.
\end{align}
\end{subequations}
A particularly nice set of stabilizer generators for the code is given in \tabref{tab: Shorstabgen}. It is clear that one logical qubit is encoded into $n=3n_r$ physical qubits.

\begin{table}[tbp]
    \centering
    \includegraphics[width=8cm]{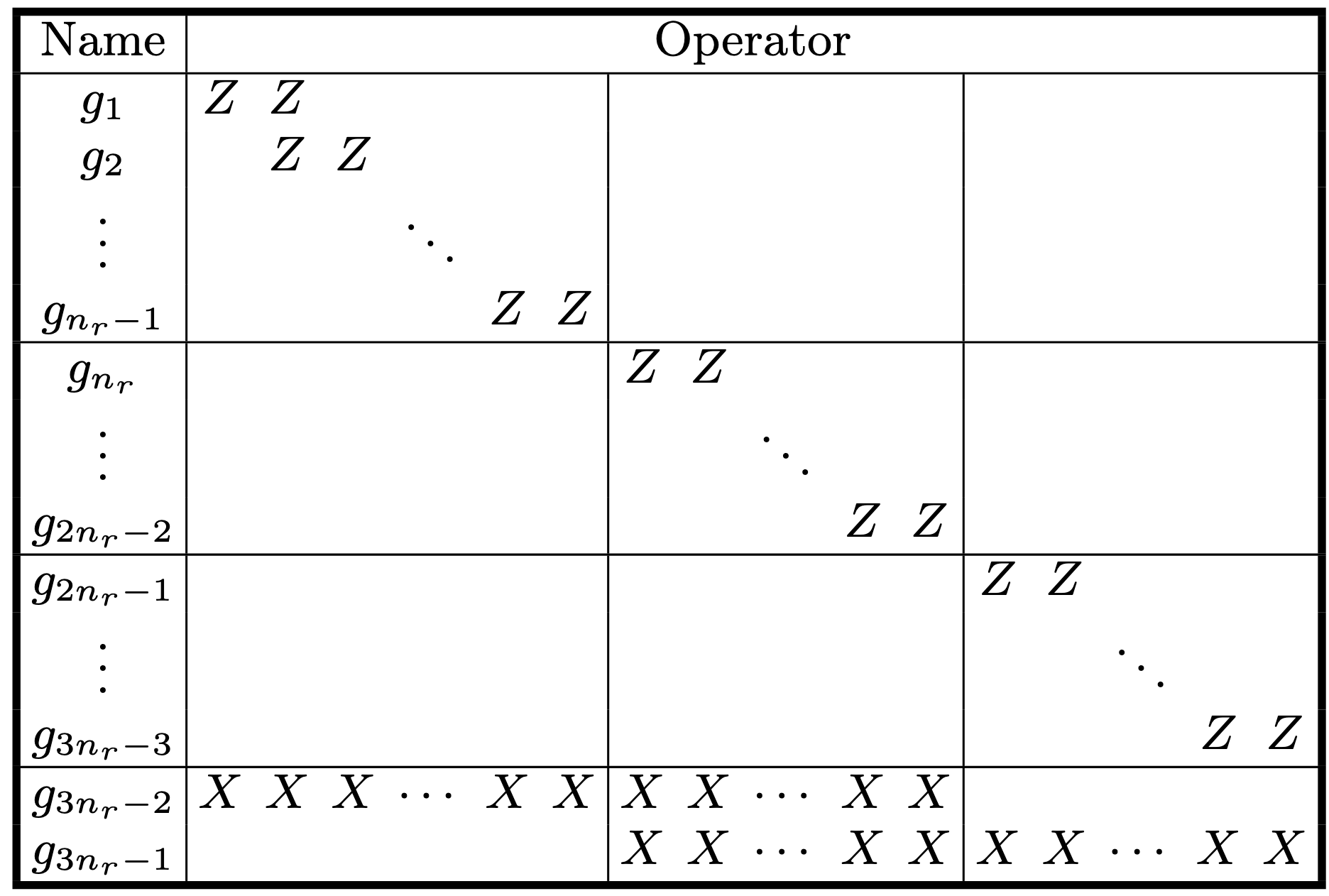}
    \caption{A set of stabilizer generators of the Shor code. The code is divided into three blocks of length $n_r$. Each block supports $(n_r-1)$ $Z$-type stabilizers, and two consecutive $Z$-type stabilizers overlap on one qubit within each block. 
    Two $Z$-type stabilizers from different blocks don't overlap.
    There are two $X$-type stabilizer generators $g_{3n_r-2}$ and $g_{3n_r-1}$, each with support on a string of $2n_r$ qubits. 
    The supports of the two $X$-type generators overlap on the middle block of qubits.
    }
    \label{tab: Shorstabgen}
\end{table}
One set of representatives of the logical $X$ and logical $Z$ operators is 
\begin{subequations}
    \begin{align}
        \overline{X}& := \left(Z\otimes I^{\otimes(n_r-1)} \right)^{\otimes 3}, \label{eqn: subeqn: ShorlogX}\\
        \overline{Z}& := X^{\otimes{n}}. \label{eqn: subeqn: ShorlogZ}
    \end{align}
\end{subequations}
Note that the logical operator $\overline{X}$ in \eqref{eqn: subeqn: ShorlogX} is a term in $G=\sum_{i<j<k}Z_iZ_jZ_k$. Thus, all $Z^{\otimes 3}$-type logical operators must be representatives of the logical $X$ operator. From the definition of the logical basis states in \eqref{eqn: Shorlogbasis}, it is clear that $Z_iZ_jZ_k$ is a logical operator if and only if the three $Z$ operators are supported on different blocks. There are $n_r^3=(n/3)^3$ such logical operators. Using \lemmaref{lemma:QFIlogicalrelation}, we have the following result.

\begin{theorem}
    The $[[n,1,3]]$ Shor code achieves a QFI of ${4n^6 t^2}/{27} = \Theta(n^6 t^2)$.
    \label{thm: shor scaling}
\end{theorem}

We also note that the above discussion on the Shor code can be generalized to estimating $Z^{\otimes k}$ terms in Hamiltonian for any constant $k \geq 3$. For example, when we choose $n_r := n/k$ (assuming $n$ is divisible by $k$) and assume $n$ is sufficiently large, then the $[[n,1,k]]$ code  
\begin{subequations}
\begin{align}
    \ket{0_\tl}=\frac{1}{\sqrt{2^k}}\left(\ket{0}^{\otimes n_r}+\ket{1}^{\otimes n_r}\right)^{\otimes k},\\
    \ket{1_\tl}=\frac{1}{\sqrt{2^k}}\left(\ket{0}^{\otimes n_r}-\ket{1}^{\otimes n_r}\right)^{\otimes k}, 
\end{align}
\end{subequations}
has $\Theta(n^k)$ $Z^{\otimes k}$-type logical operators, implying a QFI of $\Theta(n^{2k}t^2)$, which achieves the optimal SHS scaling, the same as the noiseless case.

Inspired by the repetition code substructure in the Shor code, we have a more general statement: 
\begin{theorem}
\label{thm: rep}
    Let an $[[n,1]]$ code $C$ be the concatenation of a $[[q,1]]$ stabilizer code $C'$ that has at least one $Z^{\otimes 3}$-type logical operator and can correct single-qubit Pauli-Z errors and the $[[n/q,1]]$ repetition code, where $q$ is a constant integer and $n$ is divisible by $q$. Then $C$ achieves a QFI of $\Theta(n^6 t^2)$. 
\end{theorem}

\begin{proof}
We need to prove (1)~$C$ has $\Theta(n^3)$ $Z^{\otimes 3}$-type logical operators and (2)~$C$ is an $[[n,1,3]]$ stabilizer code. 

Suppose that $Z_1 Z_2 Z_3$ is a logical operator for $C'$, it then implies $Z_i Z_j Z_k$ are logical operators for $C$ whenever $1 \leq i \leq n_r$, $n_r+1\leq j \leq 2n_r$, $2n_r+1 \leq k \leq 3n_r$ where $n_r := n/q$. Note that in $C$, qubits with labels from $q_1 n_r + 1$ to $(q_1 + 1)n_r$ correspond to the $q_1$-th qubit in $C'$ expanded from the repetition code concatenation for any $q_1$. The number of different $Z_iZ_jZ_k$ above are $n_r^3 = \Theta(n^3)$, proving statement (1). 

To prove statement (2), we need to show $C$ can correct any single-qubit error, i.e., any operator $S_iS_j$ for $S = I,X,Z$ must either belong to the stabilizer group of $C$ or anti-commute with one stabilizer of $C$. First, note that for any $i$, $X_i$ must anti-commute with $Z_iZ_k$ for any $k \neq i$ and $(\lceil i/n_r\rceil - 1) n_r + 1 \leq  k \leq \lceil i/n_r\rceil n_r $ which is a stabilizer of $C$. It implies any $S_iS_j$ anti-commutes with some stabilizer of $C$ whenever one of them is a Pauli-X operator. On the other hand, any $n$-qubit operator $S_iS_j$ with $S = I,Z$ can be mapped to a $q$-qubit operator $S_{\lceil i/n_r \rceil}S_{\lceil j/n_r \rceil}$ acting on $C'$ which either belongs to the stabilizer group of $C'$ or anti-commutes with one stabilizer of $C'$ because $C'$ corrects single-qubit Pauli-Z errors. It implies the required condition is also satisfied for $S_iS_j$ with $S = I,Z$. Statement (2) is then proven. 
\end{proof}

\thmref{thm: shor scaling} is a special case of \thmref{thm: rep} where $C'$ is the three-qubit phase-flip code. 
As discussed in \secref{sec: formulation of the problem}, the above codes provide the optimal SHS scaling we can achieve for estimating the $ZZZ$ interaction under single-qubit noise (in the setting of \eqref{eq:master} and \eqref{eq:hamt}). Note that \thmref{thm: rep} can also be straightforwardly generalized to estimating $Z^{\otimes k}$ terms for $k > 3$ if we replace $Z^{\otimes 3}$ in the statement of \thmref{thm: rep} by $Z^{\otimes k}$ and $\Theta(n^6 t^2)$ by $\Theta(n^{2k}t^2)$.

\section{Necessary Conditions for Achieving the SHS}\label{sec: no go}

In the previous section, we discussed three examples of quantum codes that asymptotically achieve $(nt)^{-1}$, $(n^2t)^{-1}$, and $(n^3t)^{-1}$ precision scalings for the $ZZZ$ interaction estimation problem. The three families of codes have special stabilizer structures: the thin surface code has a set of low-weight stabilizer generators; the QRM code lacks a complete set of low-weight stabilizer generators; and the Shor code has local low-weight $Z$-type generators even though the $X$-type generators scaled linearly with the length of the code. In this section, we address the question of what properties a stabilizer code must have to achieve different precision scalings with respect to $n$, or equivalently, different scalings of the number of $Z^{\otimes 3}$-type logical operators with respect to $n$.

\subsection{Unachievability of the SHS}\label{subsec: lin scaling}

We proved in \thmref{thm: thin planar scaling} that the thin surface code can be used to achieve $1/n$ precision scaling in $\omega$. In fact, this scaling can more generally apply to stabilizer codes with constant weight generators that can correct arbitrary single qubit errors. We want to find an upper bound on the number of $Z^{\otimes 3}$-type logical operators for such codes, which will consequently give us the best possible precision scaling.

\begin{theorem}
    Let $C$ be an $[[n,1,3]]$ stabilizer code such that there exists a set of stabilizer generators with weights at most $w$.  
    Then, $\ell$, the number of $Z^{\otimes 3}$-type logical operators of $C$ is at most $2w(w+1)n/3$. In particular, if $w$ is a constant independent of $n$, then $\ell = O(n)$. 
    \label{thm: LDPC}
\end{theorem}

\begin{proof}
Let $\sfS$ be a stabilizer generator set of $C$ such that the weight of every $g \in \sfS$ satisfies $\wt(g) \leq w$. Without loss of generality, we assume there are no $Z_i$-type stabilizers. Otherwise, we can always remove from the entire set of $n$ qubits the set of qubits on which a single Pauli-Z are stabilizers, i.e., qubits that are always in states $\ket{0}$. Then the remaining $\tilde{n}$ qubits give an $[[\tilde{n},1,3]]$ stabilizer code $\tilde{C}$ satisfying the same weight constraint without $Z_i$-type stabilizers. Any $Z^{\otimes 3}$-type logical operator in $C$ must be fully supported on the remaining $\tilde{n}$ qubits and is thus also a $Z^{\otimes 3}$-type logical operator in $\tilde{C}$. Then by proving the theorem for $\tilde{C}$, the theorem for $C$ also holds. 

We define the \textit{degree} of a qubit as the number of $Z^{\otimes 3}$-type logical operators that are supported on it. Without loss of generality, assume that the first qubit (i.e., the qubit labeled by $1$) has the highest degree. Define a graph $\frakG$ with $(n-1)$ vertices corresponding to the qubits labeled by $2,3,\ldots, n$. There is an edge between the $i$-th and $j$-th vertices if $Z_1Z_iZ_j$ is a logical operator of the code. We define the \textit{order} of a vertex as the number of edges connected to it.

Consider the two separate cases: (1) There is some vertex with order $>w$, and (2) The order of every vertex is $\leq w$.

\textbf{Case 1:} Suppose that the order of vertex $i_1$ is greater than $w$. Then, there are vertices $i_2, i_3, \ldots, i_{w+2}$ that are connected to vertex $i_1$, as shown in \figref{fig: ldpc graph}. The product of any two $Z^{\otimes 3}$-type logical operators is a stabilizer because they must be representatives of a same logical operator. It then follows that  $Z_{i_j}Z_{i_{j+1}}$ is a stabilizer for all $2\leq j \leq w+1$. Now, since $Z_{i_2}$ is not a stabilizer and must be an error operator that anti-commutes with some stabilizer generator, there is some stabilizer generator $g$ with $X$ or $Y$ acting on the qubit labeled by $i_2$. Because $g$ commutes with the stabilizer $Z_{i_2}Z_{i_3}$, it follows that $g$ has $X$ or $Y$ support on the qubit labeled by $i_3$ as well. 
\begin{figure}[tbp]
    \centering
    \includegraphics[width=0.8\linewidth]{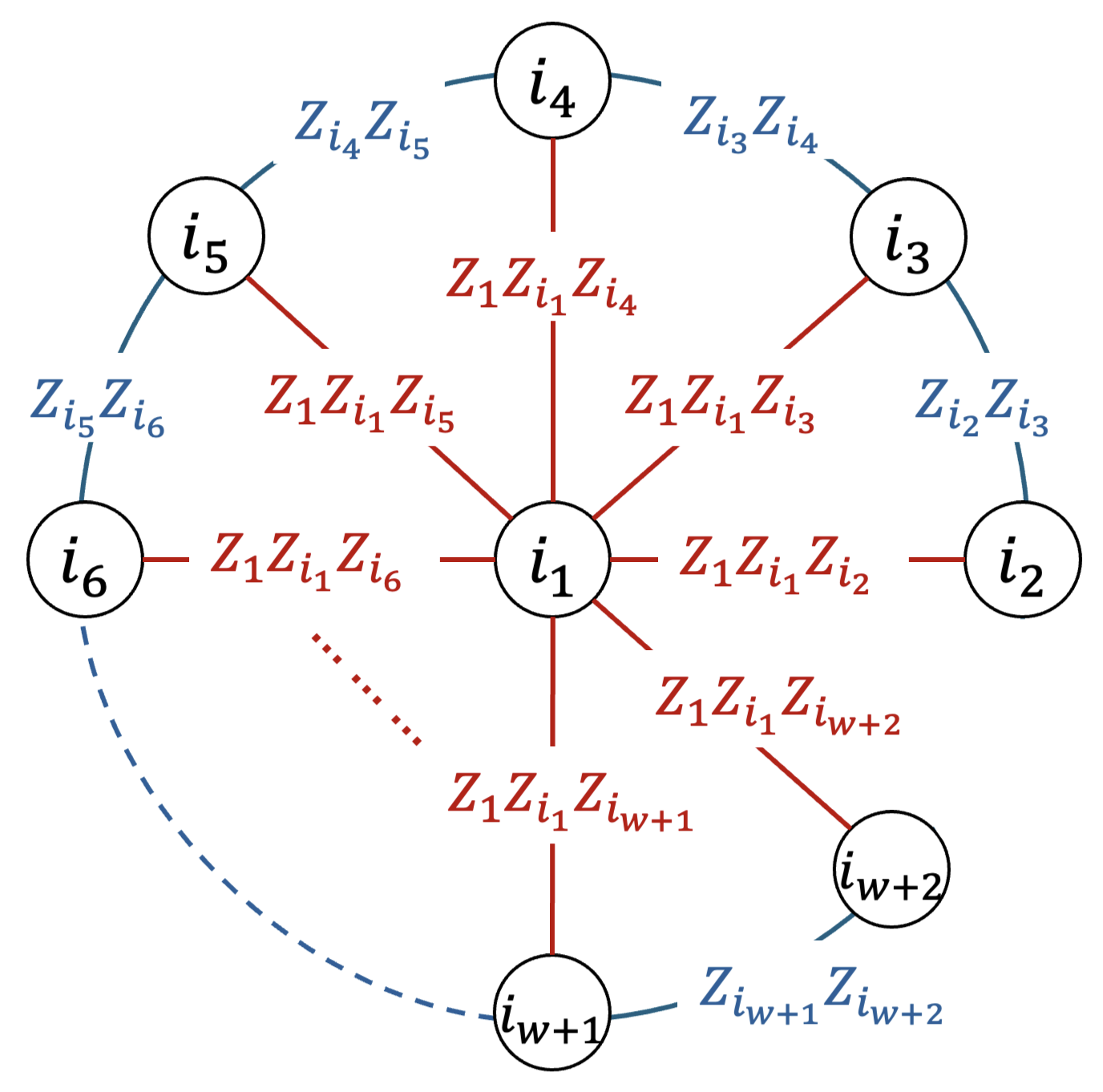}
    \caption{A subgraph of $\frakG$ where vertex $i_1$ has an order greater than $w$. Here, the red edges correspond to the $Z^{\otimes 3}$-type logical operators. Since the product of two such logical operators is a stabilizer, it follows that $Z_{i_j}Z_{i_{j+1}}$ is a stabilizer for all $2\leq j \leq w+1$. We graphically denote these stabilizers by blue lines (note that these are not edges in $\frakG$).}
    \label{fig: ldpc graph} 
\end{figure}
Similarly, we argue that $g$ must have $X$ or $Y$ support on the qubits labeled by $i_4, i_5, \ldots, i_{w+2}$. Then, $\wt(g)>w$, which is a contradiction. Hence, this case is not possible for $C$. 

\textbf{Case 2:} Suppose that the order of each vertex is $\leq w$. Let $d_1$ be the number of logical $Z^{\otimes 3}$ operators with support on the first qubit, i.e., the number of edges in $\frakG$. First, we will prove that for any subset of vertices $V$ that is a minimum \textit{vertex cover}\footnote{A vertex cover is a subset of the vertices such that at least one endpoint of each edge is in the subset. A minimum vertex cover is such a subset with the minimum number of vertices.} of $\frakG$, the size of $V$ satisfies $|V|\geq \lfloor \frac{d_1}{2w}\rfloor$.

The \textit{matching} of a graph is defined as a subset of its edges such that no two edges share a vertex. Let $M$ be a maximum matching of $\frakG$. Due to the disjointness of the edges in $M$, each vertex in $V$ covers at most one edge of the maximum matching. Therefore, $|V|\geq |M|$. So, it is sufficient to prove that $|M|\geq \lfloor \frac{d_1}{2w} \rfloor$. 
Note that for $w=1$, the edges are disjoint; hence, $|M|=d_1\geq \lfloor \frac{d_1}{2w} \rfloor$. Now, suppose $w>1$. To generate an arbitrary matching, we consider the set of all edges. First, we pick a random edge from it and remove all the edges that share endpoints with it from the set. Since the highest order of the vertices is $\leq w$, the number of edges removed is at most $2(w-1)$. We choose another random edge from the set and continue this process until the set is a matching. In the end, we will obtain a matching of size $\geq \lfloor \frac{d_1}{2(w-1) + 1} \rfloor \geq \lfloor \frac{d_1}{2w} \rfloor$. Hence, $|M|\geq \lfloor \frac{d_1}{2w} \rfloor$, which, in turn, implies that $|V| \geq \lfloor \frac{d_1}{2w} \rfloor$ for all $w$.

Now, the code must have a stabilizer generator $g$ that is $X$ or $Y$ on the first qubit because each $Z_i$ must anti-commute with some stabilizer generator. Suppose there is an edge between vertices labeled by $i$ and $j$, and so $Z_1Z_iZ_j$ is a logical operator. Since $g$ must commute with $Z_1Z_iZ_j$, it follows that $g$ has $X$ or $Y$ support on either the $i$-th or $j$-th qubit. Therefore, the vertices corresponding to the qubits that support $X$ or $Y$ of $g$ form a vertex cover of $\frakG$. Hence, $\wt(g)\geq \lfloor \frac{d_1}{2w} \rfloor$. However, by the constant weight assumption, $w \geq \wt(g)$, which implies that $\lfloor \frac{d_1}{2w} \rfloor \leq w$. So, $d_1 \leq 2 w(w+1)$. Recall that $d_1$ is the number of logical $Z^{\otimes 3}$ operators with support on the first qubit (qubit with the highest degree). Hence, the total number of $Z^{\otimes 3}$-type logical operators satisfies $\ell \leq d_1 n/3\leq 2w(w+1)n/3$. The first inequality here is from \lemmaref{lemma: d_1lowerbd} that we will prove below separately.
\end{proof}

Here we prove a lemma where we establish a simple relation between the number of $Z^{\otimes 3}$-type logical operators and the degrees of qubits, which was used in the proof above and will also be used in later sections. 
\begin{lemma}
    Suppose the number of $Z^{\otimes 3}$-type logical operators in an $n$-qubit stabilizer code is $\ell$. If $d_i$ is the degree of the $i$-th qubit of the code, then 
    \begin{align}
        \max\{d_i\}_{i=1}^n \geq \frac{3\ell}{n}.
        \label{eqn: d_1lowerbd}
    \end{align}
    \label{lemma: d_1lowerbd}
\end{lemma}
\begin{proof}
Recall that the degree of a qubit is defined to be the number of $Z^{\otimes 3}$-type logical operators that are supported on it. 
Without loss of generality, assume that the first qubit has the highest degree, i.e., $d_1=\max\{d_i\}_{i=1}^n$. The sum of the weights of all  $Z^{\otimes 3}$-type logical operators is $3\ell$. This must be equal to the sum of the degrees of the qubits. So, $\sum_{i=1}^n d_i=3 \ell $. But $\sum_{i=1}^n d_i \leq n d_1$. Hence, it follows that $nd_1\geq 3\ell$, and so $\max\{d_i\}_{i=1}^n\geq \frac{3\ell}{n}$.
\end{proof}

From \thmref{thm: LDPC} and \lemmaref{lemma:QFIlogicalrelation}, the best possible scaling of the estimation precision for the $ZZZ$ interaction using a stabilizer code with constant-weight stabilizer generators is the HS of $n^{-1}$. This family of stabilizer codes include quantum LDPC codes which are an important class of stabilizer codes as a candidate to achieve fault tolerance~\cite{breuckmann2021quantum}. The thin surface code is one class of quantum LDPC codes, and by \thmref{thm: thin planar scaling}, it obtains the optimal scaling. 

\thmref{thm: LDPC} also applies to cases where $w$ increases with respect to $n$. For example, one corollary of \thmref{thm: LDPC} is when $w = o(n)$, the optimal SHS $n^{-3}$ cannot be achieved. Note that it is also a corollary of \thmref{thm: n3scaling} that we will present later in \secref{subsec: cubic scaling}. 

\subsection{Unachievability of the Optimal SHS}\label{subsec: quadratic scaling}

We now relax the constant-weight assumption on the stabilizer generators, in which case we want to establish a different no-go result that precludes a stabilizer code from achieving a precision scaling better than $1/n^2$. 

A notable feature of the Shor code is its $Z^{\otimes 2}$-type stabilizer generators, as shown in \tabref{tab: Shorstabgen}. The following theorem demonstrates that the existence of such stabilizers is necessary for achieving the optimal SHS of $n^{-3}$.

\begin{theorem}
    Let $C$ be an $[[n,1]]$ stabilizer code. If $C$ does not have a $Z^{\otimes 2}$-type stabilizer, then the number of $Z^{\otimes 3}$-type logical operators is $\ell \leq {n(n-1)}/{6} = O(n^2)$.
    \label{thm: ZZ existence}
\end{theorem}

\begin{proof}
    Let $\mathcal{A}=\{A_1,A_2,\ldots A_\ell\}$ be the set of $Z^{\otimes 3}$-type logical operators. Let $d_i$ be the degree of the $i$-th qubit. Without loss of generality, we make the following three assumptions: (1) $d_1=\max\{d_i\}_{i=1}^n$, (2) $\mathcal{A}'=\{A_i\}_{i=1}^{d_1}$ is the set of operators with support on the first qubit, and (3) The first $n'$ qubits is the minimum set of qubits that contains the supports of operators in $\mathcal{A}'$. All three assumptions can be satisfied by relabeling the qubits. The operators in $\mathcal{A}$ are different representatives of the same logical operator, and so the product of any two $A_i,A_j\in \mathcal{A}'$ is a stabilizer. If there is no $Z^{\otimes 2}$-type stabilizer, then each pair of operators in $\mathcal{A}'$ only overlap on the first qubit. Hence, each operator in $\mathcal{A}'$ is supported on two distinct qubits other than the first qubit. Therefore, $n'=2d_1+1$. However, since $n'\leq n$, we have $2d_1+1\leq n$. Combining this inequality with \eqref{eqn: d_1lowerbd} in \lemmaref{lemma: d_1lowerbd}, we get 
    \begin{align}
        \frac{3\ell}{n}\leq d_1 \leq \frac{n-1}{2}.
        \label{eqn: d_1ineq}
    \end{align}
    Therefore, $\ell\leq \frac{n(n-1)}{6}$, and so if there are no $Z^{\otimes 2}$-type stabilizers, then $\ell=O(n^2)$. 
\end{proof}

The scaling of QRM codes derived in \thmref{thm: qrm} agrees with this result. One can show using the enumerator polynomial in \eqref{eqn: enumshortRM(1,m)perp} of \appref{app:A} that QRM codes do not have any $Z^{\otimes 2}$-type stabilizer. Therefore, by \thmref{thm: ZZ existence} and \thmref{thm: qrm}, the QRM codes achieve the optimal precision scaling $1/n^2$ among all stabilizer codes without $Z^{\otimes 2}$-type stabilizers.

Finally, we discuss generalization to the estimation of higher-order interactions $Z^{\otimes k}$ for $k \geq 3$. We claim that for any $1 \leq k_0 < k-1$, for stabilizer codes without $Z^{\otimes i}$-type stabilizers for all $i \leq 2k_0$, the number of $Z^{\otimes k}$-logical operators is at most $O(n^{k-k_0})$ which is a factor of $n^{k_0}$ away from the optimal scaling. It implies $ZZ$-type stabilizers are necessary to achieve the optimal scaling $O(n^{k})$. \thmref{thm: ZZ existence} represents the special case where $k=3$ and $k_0=1$. 

To see why the generalized statement is true, we use the generalized Ray-Chaudhuri--Wilson theorem~\cite{snevily1995generalization} which was a celebrated intersection theorem in combinatorics. Consider $\scrF$, a family of subsets of $\{1,2,\ldots,n\}$. If $\abs{F} = k$ for any $F \in \scrF$ and $\abs{E\cap F} < k - k_0$ for any distinct subsets $E,F \in \scrF$. Then the generalized Ray-Chaudhuri--Wilson theorem states that $\abs{\scrF} \leq \binom{n}{k-k_0} = \Theta(n^{k-k_0})$. In our setting, we pick $\scrF$ to be the family of subsets of qubits $\{1,2,\ldots,n\}$ such that $F \in \scrF$ means $F$ is the support of a $Z^{\otimes k}$-type logical operator. When $\abs{E\cap F} = i$, it implies the existence of a $Z^{\otimes 2(k-i)}$-type stabilizer which is equal to two $Z^{\otimes k}$-type logical operators supported on $E$ and $F$. The non-existence of $Z^{\otimes j}$-type stabilizers with $j \leq 2k_0$ implies that $2(k-\abs{E\cap F}) > 2k_0$ and that $\abs{E\cap F} < k - k_0$. Then we have an upper bound of $\Theta(n^{k-k_0})$ on $\abs{\scrF}$, i.e., the number of $Z^{\otimes k}$-type logical operators.

\subsection{Necessary Condition for Achieving the Optimal SHS}\label{subsec: cubic scaling}

As we discussed, the existence of a $Z^{\otimes 2}$-type stabilizer is necessary for achieving the optimal SHS $\sim n^{-3}$. Note that the stabilizer group of Shor codes that achieve the optimal SHS contains a chain of $Z^{\otimes 2}$-type stabilizers. Motivated by that example, we provide a necessary condition for a stabilizer code that achieves this optimal SHS. Here, we prove that if the optimal SHS is achievable by a stabilizer code, then the stabilizer group contains such a chain of linear length. 

\begin{theorem}
    Let $C_n$ be a family of $[[n,1,3]]$ stabilizer codes. Suppose that the number of $Z^{\otimes 3}$-type logical operators is $\ell=\Theta(n^3)$. Then, for all $n$, there is a subset of qubits indexed by $\{i_1,\ldots,i_{n_r}\}$ such that 
    \begin{enumerate}[wide, labelwidth=!,itemindent=!,labelindent=0pt, leftmargin=0em, label={(\arabic*)}, parsep=0pt]
        \item $\{Z_{i_1}Z_{i_2},Z_{i_2}Z_{i_3}, \ldots, Z_{i_{n_r-1}}Z_{i_{n_r}}\}$ is a subset of the stabilizers of $C_n$; 
        \item For any $1 \leq k \leq r$, $Z_{i_k}$ is not a stabilizer of $C_n$; 
        \item $n_r = \Theta(n)$.
    \end{enumerate}
    \label{thm: n3scaling}
\end{theorem}

\begin{proof} 
    Consider $C_n$. Let $\mathcal{A}=\{A_1, A_2,\ldots, A_\ell\}$ be the set of $Z^{\otimes 3}$-type logical operators. Let $d_i$ denote the degree of the $i$-th qubit. Without loss of generality, we make the same three assumptions we made in \thmref{thm: ZZ existence}: (1) $d_1=\max\{d_i\}_{i=1}^n$, (2) The set of operators with support on the first qubit is $\mathcal{A}'=\{A_i\}_{i=1}^{d_1}$, and (3) The first $n'$ qubits is the minimum set of qubits that contains the supports of operators in $\mathcal{A}'$. 

    Define a graph $\frakG$ with $(n'-1)$ vertices $\{v_2,v_3,\ldots,v_{n'}\}$ corresponding to the qubits indexed by $2,3,\ldots,n'$, where an edge connects $v_i$ and $v_j$ if $Z_1Z_iZ_j\in \mathcal{A}'$. In total, there are $d_1=|\mathcal{A}'|$ edges in the graph. Since $\ell=\Theta(n^3)$, \lemmaref{lemma: d_1lowerbd} implies that $d_1=\Omega(n^2)$. However, $\frakG$ has $(n'-1)$ vertices, and so it has at most $d_1=O(n'^2)=O(n^2)$ edges. Thus, $d_1=\Theta(n^2)$. Therefore, we can assume $\frakG$ has $(n'-1) \geq bn$ vertices with $d_1 \geq an^2$ edges for some $a,b > 0$ if we choose $n$ to be sufficiently large. 

    Note that the product of an even number of commuting logical operators ($Z^{\otimes 3}$-type operators here to be specific) in $\mathcal{A}$ is a stabilizer. Therefore, if there is an even path (a path consisting of an even number of edges) between $v_i$ and $v_j$ in $\frakG$, then $Z_iZ_j$ is a stabilizer. We will show that we can always identify a subset of vertices of size $\Theta(n)$ such that there exists an even path between any pair of adjacent vertices.

    By \lemmaref{lemma: graphlemma}, $\frakG$ has a connected component $\frakG'$ with $n_1'=\Theta(n)$ vertices and $d_1'=\Theta(n^2)$ edges. Let the number of edges in $\frakG'$ be $d_1'\geq (c/2)n_1'^2$ for some constant $c\in (0,1)$. Suppose $\frakG'$ has $\geq c_1n_1'$ vertices whose orders are $\leq c_2n_1'$, where $c_1,c_2\in (0,1)$. Then, the sum of the orders of the vertices is at most $c_1n_1'\cdot c_2n_1'+(1-c_1)n_1'\cdot n_1'$. The sum of the orders of the vertices of a graph is twice the total number of edges. Therefore, $2\cdot (c/2)n_1'^2 \leq c_1c_2n_1'^2+(1-c_1)n_1'^2$, and so 
    \begin{align}
        c\leq 1- c_1 + c_1c_2.
    \end{align}
    Note that if we choose $c_1=1-\frac{c}{2}$ and $c_2=\frac{c}{3}$, we have $c\leq 1- c_1+c_1c_2=1-\left(1-\frac{c}{2}\right)+\left(1-\frac{c}{2}\right)\frac{c}{3}<\frac{c}{2}+\frac{c}{3}<c$, which is a contradiction. Therefore, $\frakG'$ has more than $(1-c_1)n_1'=\frac{c}{2}n_1'$ vertices with orders greater than $\frac{c}{3}n_1'$.

    Reindex the vertices in descending order of their vertex orders, that is $v_2$ has the highest order, $v_3$ has the second-highest order and so on. Let $T$ be the set of the first $\lfloor\frac{c}{6}n_1'\rfloor$ vertices. Each vertex in $T$ has an order $>\frac{c}{3}n_1'$, which, in turn, is greater than $2|T|$. Therefore, for each vertex $v_i\in T$, we can assign a unique $\tilde{v}_i\not \in T$ that is connected to $v_i$ by a single edge such that $\tilde v_i \neq \tilde v_j$ for $i \neq j$. Now, starting from $i=2$, suppose that there is no even path between $v_i$ and $v_{i+1}$. Hence, all the paths from $v_i$ to $v_{i+1}$ consist of an odd number of edges. Then, by construction, there is an even path between $v_i$ and $\tilde{v}_{i+1}$, i.e., the path from $v_i$ to $v_{i+1}$ concatenated with the single edge from $v_{i+1}$ to $\tilde{v}_{i+1}$. We replace $v_{i+1}$ by $\tilde{v}_{i+1}$. Otherwise when there is already an even path between $v_i$ and $v_{i+1}$, we do no replacement. Similarly, we apply the same replacement rules for $i = 3,4,\cdots, \lfloor\frac{c}{6}n_1'\rfloor + 1$. After all the appropriate replacements, all the adjacent vertices in $T$ are connected by even paths. Therefore, each pair of adjacent qubits support a $Z^{\otimes 2}$-type stabilizer. Furthermore, each such stabilizer is independent of the others. Hence, we conclude that, by reversing the reindexings, there exists a set of qubits labeled by $\{i_1,i_2,\ldots, i_{n_r}\}$ of size $n_r = \lfloor\frac{c}{6}n_1'\rfloor = \Theta(n)$ such that $\{Z_{i_1}Z_{i_2},Z_{i_2}Z_{i_3}, \ldots, Z_{i_{n_r-1}}Z_{i_{n_r}}\}$ is a subset of stabilizers. 

    From the discussion above, for any $1 \leq k \leq r$, there exists some qubit $j_k$ such that $Z_1 Z_{i_k} Z_{j_k}$ is a logical operator in $C_n$ because $Z_{i_k}$ must belong to some connected component of $\frakG$. It implies $Z_{i_k}$ is not a stabilizer of $C_n$ because otherwise $Z_1 Z_{j_k} = Z_1 Z_{i_k} Z_{j_k} Z_{i_k}$ is a logical operator of $C_n$, contradicting with the fact that the code distance is $3$. 
\end{proof}

Note that having $\{Z_{i_1}Z_{i_2},\ldots,Z_{i_{n_r-1}}Z_{i_{n_r}}\}$ as a part of the stabilizers of $C$ guarantees the codewords on qubits $\{i_1,\ldots,i_{n_r}\}$ must be in the span of $\{\ket{0^{\otimes n_r}},\ket{1^{\otimes n_r}}\}$. Furthermore, as $\{Z_{i_k}\}_{1 \leq k \leq r}$ are not stabilizers, for any choice of stabilizer generators there must be one of them that acts non-trivially as $X$ or $Y$ on all qubits in $\{i_1,\ldots,i_{n_r}\}$. It implies the necessary existence of a weight-$\Theta(n)$ stabilizer generator when the optimal SHS is achieved. 
\thmref{thm: n3scaling} shows the repetition code substructure described above is necessary in achieving the optimal SHS. 

\section{Conclusion and Outlook}\label{sec: conclusion}

In conclusion, we studied how different stabilizer codes perform when estimating the parameter $\omega$ encoded in the 3-local Hamiltonian in \eqref{eq:hamt} under single-qubit noise. We also derived what properties stabilizer codes must satisfy in order to asymptotically achieve the SHS or the optimal SHS. Our work suggests more investigation towards the direction of exploiting the stabilizer codes for many-body quantum metrology. 

The crux of the solution of our problem was the relation between the QFI and the number of logical $Z^{\otimes 3}$-type operators of the stabilizer code, which connects quantum metrology to the properties of stabilizer codes. Based on it, some of our results can be straightforwardly generalized to probing higher-order interactions under stronger quantum noise. For example, the $[[n,1,3]]$ thin surface code, QRM code and Shor code discussed in \secref{sec:achieving} all have natural generalizations to $[[n,1,d>3]]$ codes, that may be analogously applied in estimating higher-order interactions.

Many interesting questions remain open. For example, our results assumed the achievement of the HL, based on which we analyzed the achievability of the SHS. It would be interesting to consider whether better scalings with respect to the number of qubits will be available if we abandon the requirement of the HL. Another question is to ask how QEC protocols behave in many-body Hamiltonian estimation when taking fault-tolerance into consideration (see e.g., \cite{kapourniotis2019fault}). Our current analysis was based on fast and precise QEC, and including measurement noise or circuit-level noise is left for future work. Our approach only considers the estimation problem for a single parameter; however we anticipate the study of QEC-assisted multi-parameter estimation as a future direction.

\begin{acknowledgements}
S.Z. thanks Victor V. Albert and Michael Vasmer for helpful discussions. S.A. and S.Z. acknowledge funding provided by Perimeter Institute for Theoretical Physics, a research institute supported in part by the Government of Canada through the Department of Innovation, Science and Economic Development Canada and by the Province of Ontario through the Ministry of Colleges and Universities. 
\end{acknowledgements}

\bibliographystyle{unsrt}
\bibliography{reference.bib}

\begin{thebibliography}{10}

\bibitem{giovannetti2011advances}
Vittorio Giovannetti, Seth Lloyd, and Lorenzo Maccone.
\newblock \href{https://doi.org/10.1038/nphoton.2011.35}{Advances in quantum metrology}.
\newblock {\em Nat. Photonics.}, 5(4):222, 2011.

\bibitem{degen2017quantum}
C.~L. Degen, F.~Reinhard, and P.~Cappellaro.
\newblock \href{https://doi.org/10.1103/RevModPhys.89.035002}{Quantum sensing}.
\newblock {\em Rev. Mod. Phys.}, 89(3):035002, Jul 2017.

\bibitem{pezze2018quantum}
Luca Pezz\`e, Augusto Smerzi, Markus~K. Oberthaler, Roman Schmied, and Philipp Treutlein.
\newblock \href{https://doi.org/10.1103/RevModPhys.90.035005}{Quantum metrology with nonclassical states of atomic ensembles}.
\newblock {\em Rev. Mod. Phys.}, 90(3):035005, Sep 2018.

\bibitem{pirandola2018advances}
Stefano Pirandola, Bhaskar~Roy Bardhan, Tobias Gehring, Christian Weedbrook, and Seth Lloyd.
\newblock \href{https://doi.org/10.1038/s41566-018-0301-6}{Advances in photonic quantum sensing}.
\newblock {\em Nat. Photonics.}, 12(12):724, 2018.

\bibitem{jiao2023quantum}
Lin Jiao, Wei Wu, Si-Yuan Bai, and Jun-Hong An.
\newblock \href{https://doi.org/10.1002/qute.202300218}{Quantum Metrology in the Noisy Intermediate-Scale Quantum Era}.
\newblock {\em Advanced Quantum Technologies}, page 2300218, 2023.

\bibitem{giovannetti2004quantum}
Vittorio Giovannetti, Seth Lloyd, and Lorenzo Maccone.
\newblock \href{https://doi.org/10.1126/science.1104149}{Quantum-Enhanced Measurements: Beating the Standard Quantum Limit}.
\newblock {\em Science}, 306(5700):1330--1336, 2004.

\bibitem{giovannetti2006quantum}
Vittorio Giovannetti, Seth Lloyd, and Lorenzo Maccone.
\newblock \href{https://doi.org/10.1103/PhysRevLett.96.010401}{Quantum Metrology}.
\newblock {\em Phys. Rev. Lett.}, 96(1):010401, Jan 2006.

\bibitem{huelga1997improvement}
S.~F. Huelga, C.~Macchiavello, T.~Pellizzari, A.~K. Ekert, M.~B. Plenio, and J.~I. Cirac.
\newblock \href{https://doi.org/10.1103/PhysRevLett.79.3865}{Improvement of Frequency Standards with Quantum Entanglement}.
\newblock {\em Phys. Rev. Lett.}, 79(20):3865--3868, Nov 1997.

\bibitem{escher2011general}
B.~M. Escher, R.~L. de~Matos~Filho, and L~Davidovich.
\newblock \href{https://doi.org/10.1038/nphys1958}{General framework for estimating the ultimate precision limit in noisy quantum-enhanced metrology}.
\newblock {\em Nat. Phys.}, 7(5):406, 2011.

\bibitem{demkowicz2012elusive}
Rafa{\l} Demkowicz-Dobrza{\'n}ski, Jan Ko{\l}ody{\'n}ski, and M{\u{a}}d{\u{a}}lin Gu{\c{t}}{\u{a}}.
\newblock \href{https://doi.org/10.1038/ncomms2067}{The elusive Heisenberg limit in quantum-enhanced metrology}.
\newblock {\em Nat. Commun.}, 3:1063, 2012.

\bibitem{demkowicz2017adaptive}
Rafa{\l} Demkowicz-Dobrza{\'n}ski, Jan Czajkowski, and Pavel Sekatski.
\newblock \href{https://doi.org/10.1103/PhysRevX.7.041009}{Adaptive Quantum Metrology under General Markovian Noise}.
\newblock {\em Phys. Rev. X}, 7(4):041009, Oct 2017.

\bibitem{zhou2018achieving}
Sisi Zhou, Mengzhen Zhang, John Preskill, and Liang Jiang.
\newblock \href{https://doi.org/10.1038/s41467-017-02510-3}{Achieving the Heisenberg limit in quantum metrology using quantum error correction}.
\newblock {\em Nat. Commun.}, 9(1):78, 2018.

\bibitem{kessler2014quantum}
E.~M. Kessler, I.~Lovchinsky, A.~O. Sushkov, and M.~D. Lukin.
\newblock \href{https://doi.org/10.1103/PhysRevLett.112.150802}{Quantum Error Correction for Metrology}.
\newblock {\em Phys. Rev. Lett.}, 112(15):150802, Apr 2014.

\bibitem{arrad2014increasing}
G.~Arrad, Y.~Vinkler, D.~Aharonov, and A.~Retzker.
\newblock \href{https://doi.org/10.1103/PhysRevLett.112.150801}{Increasing Sensing Resolution with Error Correction}.
\newblock {\em Phys. Rev. Lett.}, 112(15):150801, Apr 2014.

\bibitem{dur2014improved}
W.~D\"ur, M.~Skotiniotis, F.~Fr\"owis, and B.~Kraus.
\newblock \href{https://doi.org/10.1103/PhysRevLett.112.080801}{Improved Quantum Metrology Using Quantum Error Correction}.
\newblock {\em Phys. Rev. Lett.}, 112(8):080801, Feb 2014.

\bibitem{sekatski2017quantum}
Pavel Sekatski, Michalis Skotiniotis, Janek Ko{\l{}}ody{\'{n}}ski, and Wolfgang D{\"{u}}r.
\newblock \href{https://doi.org/10.22331/q-2017-09-06-27}{Quantum metrology with full and fast quantum control}.
\newblock {\em {Quantum}}, 1:27, September 2017.

\bibitem{layden2019ancilla}
David Layden, Sisi Zhou, Paola Cappellaro, and Liang Jiang.
\newblock \href{https://doi.org/10.1103/PhysRevLett.122.040502}{Ancilla-Free Quantum Error Correction Codes for Quantum Metrology}.
\newblock {\em Phys. Rev. Lett.}, 122(4):040502, Jan 2019.

\bibitem{zhou2020optimal}
Sisi Zhou and Liang Jiang.
\newblock \href{https://doi.org/10.1103/PhysRevResearch.2.013235}{Optimal approximate quantum error correction for quantum metrology}.
\newblock {\em Physical Review Research}, 2(1):013235, 2020.

\bibitem{zhou2021asymptotic}
Sisi Zhou and Liang Jiang.
\newblock \href{https://doi.org/10.1103/PRXQuantum.2.010343}{Asymptotic theory of quantum channel estimation}.
\newblock {\em PRX Quantum}, 2(1):010343, 2021.

\bibitem{zhou2024achieving}
Sisi Zhou, Argyris~Giannisis Manes, and Liang Jiang.
\newblock \href{https://doi.org/10.1103/PhysRevA.109.042406}{Achieving metrological limits using ancilla-free quantum error-correcting codes}.
\newblock {\em Physical Review A}, 109(4):042406, 2024.

\bibitem{gottesman1997stabilizer}
Daniel Gottesman.
\newblock {\em \href{https://doi.org/10.7907/rzr7-dt72}{Stabilizer codes and quantum error correction}}.
\newblock California Institute of Technology, 1997.

\bibitem{nielsen2001quantum}
Michael~A Nielsen and Isaac~L Chuang.
\newblock {\em \href{https://doi.org/10.1017/CBO9780511976667}{Quantum computation and quantum information}}, volume~2.
\newblock Cambridge University Press Cambridge, 2001.

\bibitem{wineland1992spin}
D.~J. Wineland, J.~J. Bollinger, W.~M. Itano, F.~L. Moore, and D.~J. Heinzen.
\newblock \href{https://doi.org/10.1103/PhysRevA.46.R6797}{Spin squeezing and reduced quantum noise in spectroscopy}.
\newblock {\em Phys. Rev. A}, 46(11):R6797--R6800, Dec 1992.

\bibitem{bollinger1996optimal}
J.~J. Bollinger, Wayne~M. Itano, D.~J. Wineland, and D.~J. Heinzen.
\newblock \href{https://doi.org/10.1103/PhysRevA.54.R4649}{Optimal frequency measurements with maximally correlated states}.
\newblock {\em Phys. Rev. A}, 54:R4649--R4652, Dec 1996.

\bibitem{leibfried2004toward}
D.~Leibfried, M.~D. Barrett, T.~Schaetz, J.~Britton, J.~Chiaverini, W.~M. Itano, J.~D. Jost, C.~Langer, and D.~J. Wineland.
\newblock \href{https://doi.org/10.1126/science.1097576}{Toward Heisenberg-Limited Spectroscopy with Multiparticle Entangled States}.
\newblock {\em Science}, 304(5676):1476--1478, 2004.

\bibitem{higgins2007entanglement}
Brendon~L Higgins, Dominic~W Berry, Stephen~D Bartlett, Howard~M Wiseman, and Geoff~J Pryde.
\newblock \href{https://doi.org/10.1038/nature06257}{Entanglement-free Heisenberg-limited phase estimation}.
\newblock {\em Nature}, 450(7168):393, 2007.

\bibitem{kaubruegger2021quantum}
Raphael Kaubruegger, Denis~V Vasilyev, Marius Schulte, Klemens Hammerer, and Peter Zoller.
\newblock \href{https://doi.org/10.1103/PhysRevX.11.041045}{Quantum variational optimization of Ramsey interferometry and atomic clocks}.
\newblock {\em Physical Review X}, 11(4):041045, 2021.

\bibitem{huang2023learning}
Hsin-Yuan Huang, Yu~Tong, Di~Fang, and Yuan Su.
\newblock \href{https://doi.org/10.1103/PhysRevLett.130.200403}{Learning many-body Hamiltonians with Heisenberg-limited scaling}.
\newblock {\em Physical Review Letters}, 130(20):200403, 2023.

\bibitem{dutkiewicz2023advantage}
Alicja Dutkiewicz, Thomas~E O'Brien, and Thomas Schuster.
\newblock \href{https://doi.org/10.22331/q-2024-11-26-1537}{The advantage of quantum control in many-body Hamiltonian learning}.
\newblock {\em Quantum}, 8:1537, 2024.

\bibitem{boixo2007generalized}
Sergio Boixo, Steven~T Flammia, Carlton~M Caves, and John~M Geremia.
\newblock \href{https://doi.org/10.1103/PhysRevLett.98.090401}{Generalized limits for single-parameter quantum estimation}.
\newblock {\em Physical Review Letters}, 98(9):090401, 2007.

\bibitem{beau2017nonlinear}
Mathieu Beau and Adolfo del Campo.
\newblock \href{https://doi.org/10.1103/PhysRevLett.119.010403}{Nonlinear quantum metrology of many-body open systems}.
\newblock {\em Physical Review Letters}, 119(1):010403, 2017.

\bibitem{czajkowski2019many}
Jan Czajkowski, Krzysztof Paw{\l}owski, and Rafa{\l} Demkowicz-Dobrza{\'n}ski.
\newblock \href{https://doi.org/10.1088/1367-2630/ab1fc2}{Many-body effects in quantum metrology}.
\newblock {\em New Journal of Physics}, 21(5):053031, 2019.

\bibitem{boixo2008quantum}
Sergio Boixo, Animesh Datta, Steven~T Flammia, Anil Shaji, Emilio Bagan, and Carlton~M Caves.
\newblock \href{https://doi.org/10.1103/PhysRevA.77.012317}{Quantum-limited metrology with product states}.
\newblock {\em Physical Review A—Atomic, Molecular, and Optical Physics}, 77(1):012317, 2008.

\bibitem{kitaev2003fault}
A~Yu Kitaev.
\newblock \href{https://doi.org/10.1016/S0003-4916(02)00018-0}{Fault-tolerant quantum computation by anyons}.
\newblock {\em Annals of physics}, 303(1):2--30, 2003.

\bibitem{dennis2002topological}
Eric Dennis, Alexei Kitaev, Andrew Landahl, and John Preskill.
\newblock \href{https://doi.org/10.1063/1.1499754}{Topological quantum memory}.
\newblock {\em Journal of Mathematical Physics}, 43(9):4452--4505, 2002.

\bibitem{fowler2012surface}
Austin~G Fowler, Matteo Mariantoni, John~M Martinis, and Andrew~N Cleland.
\newblock \href{https://doi.org/10.1103/PhysRevA.86.032324}{Surface codes: Towards practical large-scale quantum computation}.
\newblock {\em Physical Review A—Atomic, Molecular, and Optical Physics}, 86(3):032324, 2012.

\bibitem{steane1999quantum}
Andrew~M Steane.
\newblock \href{https://doi.org/10.1109/18.771249}{Quantum reed-muller codes}.
\newblock {\em IEEE Transactions on Information Theory}, 45(5):1701--1703, 1999.

\bibitem{shor1995scheme}
Peter~W Shor.
\newblock \href{https://doi.org/10.1103/PhysRevA.52.R2493}{Scheme for reducing decoherence in quantum computer memory}.
\newblock {\em Physical Review A}, 52(4):R2493, 1995.

\bibitem{breuckmann2021quantum}
Nikolas~P Breuckmann and Jens~Niklas Eberhardt.
\newblock \href{https://doi.org/10.1103/PRXQuantum.2.040101}{Quantum low-density parity-check codes}.
\newblock {\em PRX Quantum}, 2(4):040101, 2021.

\bibitem{granade2012robust}
Christopher~E Granade, Christopher Ferrie, Nathan Wiebe, and David~G Cory.
\newblock \href{https://doi.org/10.1088/1367-2630/14/10/103013}{Robust online Hamiltonian learning}.
\newblock {\em New Journal of Physics}, 14(10):103013, 2012.

\bibitem{hincks2018hamiltonian}
Ian Hincks, Thomas Alexander, Michal Kononenko, Benjamin Soloway, and David~G Cory.
\newblock \href{https://doi.org/10.48550/arXiv.1806.02427}{Hamiltonian learning with online Bayesian experiment design in practice}.
\newblock {\em arXiv:1806.02427}, 2018.

\bibitem{wiebe2014hamiltonian}
Nathan Wiebe, Christopher Granade, Christopher Ferrie, and David~G Cory.
\newblock \href{https://doi.org/10.1103/PhysRevLett.112.190501}{Hamiltonian learning and certification using quantum resources}.
\newblock {\em Physical Review Letters}, 112(19):190501, 2014.

\bibitem{evans2019scalable}
Tim~J Evans, Robin Harper, and Steven~T Flammia.
\newblock \href{https://doi.org/10.48550/arXiv.1912.07636}{Scalable bayesian hamiltonian learning}.
\newblock {\em arXiv:1912.07636}, 2019.

\bibitem{yu2023robust}
Wenjun Yu, Jinzhao Sun, Zeyao Han, and Xiao Yuan.
\newblock \href{https://doi.org/10.22331/q-2023-06-29-1045}{Robust and efficient Hamiltonian learning}.
\newblock {\em Quantum}, 7:1045, 2023.

\bibitem{li2020hamiltonian}
Zhi Li, Liujun Zou, and Timothy~H Hsieh.
\newblock \href{https://doi.org/10.1103/PhysRevLett.124.160502}{Hamiltonian tomography via quantum quench}.
\newblock {\em Physical Review Letters}, 124(16):160502, 2020.

\bibitem{hangleiter2024robustly}
Dominik Hangleiter, Ingo Roth, Jon{\'a}{\v{s}} Fuksa, Jens Eisert, and Pedram Roushan.
\newblock \href{https://doi.org/10.1038/s41467-024-52629-3}{Robustly learning the Hamiltonian dynamics of a superconducting quantum processor}.
\newblock {\em Nature Communications}, 15(1):9595, 2024.

\bibitem{stilck2024efficient}
Daniel Stilck~Fran{\c{c}}a, Liubov~A Markovich, V.~V. Dobrovitski, Albert~H. Werner, and Johannes Borregaard.
\newblock \href{https://doi.org/10.1038/s41467-023-44012-5}{Efficient and robust estimation of many-qubit Hamiltonians}.
\newblock {\em Nature Communications}, 15(1):311, 2024.

\bibitem{anshu2020sample}
Anurag Anshu, Srinivasan Arunachalam, Tomotaka Kuwahara, and Mehdi Soleimanifar.
\newblock \href{https://doi.org/10.1109/FOCS46700.2020.00069}{Sample-efficient learning of quantum many-body systems}.
\newblock In {\em 2020 IEEE 61st Annual Symposium on Foundations of Computer Science (FOCS)}, pages 685--691. IEEE, 2020.

\bibitem{haah2022optimal}
Jeongwan Haah, Robin Kothari, and Ewin Tang.
\newblock \href{https://doi.org/10.1109/FOCS54457.2022.00020}{Optimal learning of quantum Hamiltonians from high-temperature Gibbs states}.
\newblock In {\em 2022 IEEE 63rd Annual Symposium on Foundations of Computer Science (FOCS)}, pages 135--146. IEEE, 2022.

\bibitem{bairey2019learning}
Eyal Bairey, Itai Arad, and Netanel~H Lindner.
\newblock \href{https://doi.org/10.1103/PhysRevLett.122.020504}{Learning a local Hamiltonian from local measurements}.
\newblock {\em Physical Review Letters}, 122(2):020504, 2019.

\bibitem{qi2019determining}
Xiao-Liang Qi and Daniel Ranard.
\newblock \href{https://doi.org/10.22331/q-2019-07-08-159}{Determining a local Hamiltonian from a single eigenstate}.
\newblock {\em Quantum}, 3:159, 2019.

\bibitem{jarzyna2015true}
Marcin Jarzyna and Rafa{\l} Demkowicz-Dobrza{\'n}ski.
\newblock \href{https://doi.org/10.1088/1367-2630/17/1/013010}{True precision limits in quantum metrology}.
\newblock {\em New Journal of Physics}, 17(1):013010, 2015.

\bibitem{rubio2018non}
Jes{\'u}s Rubio, Paul Knott, and Jacob Dunningham.
\newblock \href{https://doi.org/10.1088/2399-6528/aaa234}{Non-asymptotic analysis of quantum metrology protocols beyond the Cram{\'e}r--Rao bound}.
\newblock {\em Journal of Physics Communications}, 2(1):015027, 2018.

\bibitem{gorini1976completely}
Vittorio Gorini, Andrzej Kossakowski, and Ennackal Chandy~George Sudarshan.
\newblock \href{https://doi.org/10.1063/1.522979}{Completely positive dynamical semigroups of N-level systems}.
\newblock {\em Journal of Mathematical Physics}, 17(5):821--825, 1976.

\bibitem{lindblad1976generators}
Goran Lindblad.
\newblock \href{https://doi.org/10.1007/BF01608499}{On the generators of quantum dynamical semigroups}.
\newblock {\em Communications in Mathematical Physics}, 48:119--130, 1976.

\bibitem{breuer2002theory}
Heinz-Peter Breuer and Francesco Petruccione.
\newblock {\em \href{https://doi.org/10.1093/acprof:oso/9780199213900.001.0001}{The theory of open quantum systems}}.
\newblock OUP Oxford, 2002.

\bibitem{holevo2011probabilistic}
Alexander~S Holevo.
\newblock {\em \href{https://doi.org/10.1007/978-88-7642-378-9}{Probabilistic and statistical aspects of quantum theory}}, volume~1.
\newblock Springer Science \& Business Media, 2011.

\bibitem{helstrom1969quantum}
Carl~W Helstrom.
\newblock \href{https://doi.org/10.1007/BF01007479}{Quantum detection and estimation theory}.
\newblock {\em Journal of Statistical Physics}, 1:231--252, 1969.

\bibitem{braunstein1994statistical}
Samuel~L Braunstein and Carlton~M Caves.
\newblock \href{https://doi.org/10.1103/PhysRevLett.72.3439}{Statistical distance and the geometry of quantum states}.
\newblock {\em Physical Review Letters}, 72(22):3439, 1994.

\bibitem{paris2009quantum}
Matteo~GA Paris.
\newblock \href{https://doi.org/10.1142/S0219749909004839}{Quantum estimation for quantum technology}.
\newblock {\em Int. J. Quantum Inf.}, 7(supp01):125--137, 2009.

\bibitem{paris2004quantum}
Matteo Paris and Jaroslav Rehacek.
\newblock {\em \href{https://doi.org/10.1007/b98673}{Quantum state estimation}}, volume 649.
\newblock Springer Science \& Business Media, 2004.

\bibitem{bennett1996mixed}
Charles~H Bennett, David~P DiVincenzo, John~A Smolin, and William~K Wootters.
\newblock \href{https://doi.org/10.1103/PhysRevA.54.3824}{Mixed-state entanglement and quantum error correction}.
\newblock {\em Physical Review A}, 54(5):3824, 1996.

\bibitem{knill1997theory}
Emanuel Knill and Raymond Laflamme.
\newblock \href{https://doi.org/10.1103/PhysRevA.55.900}{Theory of quantum error-correcting codes}.
\newblock {\em Physical Review A}, 55(2):900, 1997.

\bibitem{google2023suppressing}
Google~Quantum AI.
\newblock \href{https://doi.org/10.1038/s41586-022-05434-1}{Suppressing quantum errors by scaling a surface code logical qubit}.
\newblock {\em Nature}, 614(7949):676--681, 2023.

\bibitem{ryan2021realization}
Ciaran Ryan-Anderson, Justin~G Bohnet, Kenny Lee, Daniel Gresh, Aaron Hankin, J. P. Gaebler, David Francois, Alexander Chernoguzov, Dominic Lucchetti, Natalie C. Brown, et~al.
\newblock \href{https://doi.org/10.1103/PhysRevX.11.041058}{Realization of real-time fault-tolerant quantum error correction}.
\newblock {\em Physical Review X}, 11(4):041058, 2021.

\bibitem{bluvstein2024logical}
Dolev Bluvstein, Simon~J Evered, Alexandra~A Geim, Sophie~H Li, Hengyun Zhou, Tom Manovitz, Sepehr Ebadi, Madelyn Cain, Marcin Kalinowski, Dominik Hangleiter, et~al.
\newblock \href{https://doi.org/10.1038/s41586-023-06927-3}{Logical quantum processor based on reconfigurable atom arrays}.
\newblock {\em Nature}, 626(7997):58--65, 2024.

\bibitem{calderbank1996good}
A~Robert Calderbank and Peter~W Shor.
\newblock \href{https://doi.org/10.1103/PhysRevA.54.1098}{Good quantum error-correcting codes exist}.
\newblock {\em Physical Review A}, 54(2):1098, 1996.

\bibitem{steane1996error}
Andrew~M Steane.
\newblock \href{https://doi.org/10.1103/PhysRevLett.77.793}{Error correcting codes in quantum theory}.
\newblock {\em Physical Review Letters}, 77(5):793, 1996.

\bibitem{anderson2014fault}
Jonas~T Anderson, Guillaume Duclos-Cianci, and David Poulin.
\newblock \href{https://doi.org/10.1103/PhysRevLett.113.080501}{Fault-tolerant conversion between the steane and reed-muller quantum codes}.
\newblock {\em Physical Review Letters}, 113(8):080501, 2014.

\bibitem{macwilliams1977theory}
Florence~Jessie MacWilliams and Neil James~Alexander Sloane.
\newblock {\em \href{https://doi.org/10.1002/9781118032749}{The theory of error-correcting codes}}, volume~16.
\newblock Elsevier, 1977.

\bibitem{zeng2011transversality}
Bei Zeng, Andrew Cross, and Isaac~L Chuang.
\newblock \href{https://doi.org/10.1109/TIT.2011.2161917}{Transversality versus universality for additive quantum codes}.
\newblock {\em IEEE Transactions on Information Theory}, 57(9):6272--6284, 2011.

\bibitem{snevily1995generalization}
Hunter~S Snevily.
\newblock \href{https://doi.org/10.1002/jcd.3180030505}{A generalization of the ray-chaudhuri-wilson theorem}.
\newblock {\em Journal of Combinatorial Designs}, 3(5):349--352, 1995.

\bibitem{kapourniotis2019fault}
Theodoros Kapourniotis and Animesh Datta.
\newblock \href{https://doi.org/10.1103/PhysRevA.100.022335}{Fault-tolerant quantum metrology}.
\newblock {\em Physical Review A}, 100(2):022335, 2019.

\end{thebibliography}

\appendix
\onecolumn 

\setcounter{theorem}{0}
\setcounter{proposition}{0}
\setcounter{lemma}{0}
\setcounter{figure}{0}
\renewcommand{\thefigure}{S\arabic{figure}}
\renewcommand{\thelemma}{S\arabic{lemma}}
\renewcommand{\thetheorem}{S\arabic{theorem}}
\renewcommand{\thecorollary}{S\arabic{corollary}}
\renewcommand{\theproposition}{S\arabic{proposition}}
\renewcommand{\theHfigure}{Supplement.\arabic{figure}}
\renewcommand{\theHlemma}{Supplement.\arabic{lemma}}
\renewcommand{\theHtheorem}{Supplement.\arabic{theorem}}
\renewcommand{\theHcorollary}{Supplement.\arabic{corollary}}

\section{Optimal QFI Coefficient}
\label{app:opt-qfi}

Here we compute the optimal QFI coefficient from \eqref{eq:opt-coeff} 
that provides the optimal coefficient of the HL among all possible QEC protocols. 
First, note that 
\begin{align}
    \frac{(\Delta G_\eff)_\opt}{2} &= \min_{\text{2-local } S} \bigg\|\sum_{i<j<k} Z_iZ_jZ_k - S\bigg\|\\
    &= \min_{\alpha_{ij},\beta_i,\gamma \in \bR} \bigg\|\sum_{i<j<k} Z_iZ_jZ_k - \sum_{i<j} \alpha_{ij} Z_iZ_j - \sum_{i} \beta_i Z_i - \gamma I \bigg\|, 
\end{align}
because any off-diagonal term (in the computational basis) will only increase the operator norm. Next, let $A(\alpha,\beta,\gamma) := \sum_{i<j<k} Z_iZ_jZ_k - \sum_{i<j} \alpha_{ij} Z_iZ_j - \sum_{i} \beta_i Z_i - \gamma I$, where we use $\alpha$ (or $\beta$) to describe the column vector whose elements are $\alpha_{ij}$ (or $\beta_i$). We have $X^{\otimes n} A(\alpha,\beta,\gamma) X^{\otimes n} = - A(-\alpha,\beta,-\gamma) $. Thus, $\norm{2 A(0,\beta,0)} = \norm{A(\alpha,\beta,\gamma) + A(-\alpha,\beta,-\gamma)} \leq \norm{A(\alpha,\beta,\gamma)} + \norm{A(-\alpha,\beta,-\gamma)} = 2 \norm{A(\alpha,\beta,\gamma)}$, where we use the triangle inequality and the unitary invariance of operator norm. It implies that
\begin{equation}
   \frac{(\Delta G_\eff)_\opt}{2} = \min_{\beta \in \bR^n} \norm{A(0,\beta,0)} . 
\end{equation}
Furthermore, let $\mP$ represents the set of all permutations (of qubits) and $P$ its element, we have $\norm{A(0,\avg(\beta),0)} \leq \frac{1}{\abs{\mP}} \sum_{P\in\mP} \norm{A(0,P\beta,0)} = \norm{A(0,\beta,0)}$, where $\avg(\beta) = (\bar\beta,\ldots,\bar\beta)$, $\bar\beta = \frac{1}{n}\sum_i \beta_i$, $P\beta = (\beta_{P(1)},\beta_{P(2)},\ldots,\beta_{P(n)})$ and we use $\norm{A(0,P\beta,0)} = \norm{A(0,\beta,0)}$ for any permutation $P$. It implies 
\begin{align}
   \frac{(\Delta G_\eff)_\opt}{2} = \min_{\bar\beta \in \bR} \norm{A(0,\avg(\beta),0)}
   = \min_{\bar\beta \in \bR} \bigg\|\sum_{i<j<k} Z_iZ_jZ_k - \bar\beta \sum_{i} Z_i \bigg\|. 
\end{align}
For any computational state $\ket{\phi(k)}$ that contains $k$ $\ket{1}$s and $(n-k)$ $\ket{0}$s, we have $A(0,\avg(\beta),0) \ket{\phi(k)} = \mu_k\ket{\phi(k)}$, where $\mu_k = \big( - \binom{k}{3} + \binom{k}{2}\binom{n-k}{1} - \binom{k}{1}\binom{n-k}{2} + \binom{n-k}{3} - \bar\beta (n-2k) \big)$. So,
\begin{equation}
    \mu_{k+1}-\mu_k =  -4k^2 + 4(n-1)k + ( 2\bar\beta - (n-2)(n-1)). 
\end{equation}
When $2\bar\beta  \leq - (n-1)$, the value of $\mu_k$ increases as $k$ increases, and the maximum and minimum eigenvalues are taken at $k=0$ and $k=n$. When $2\bar\beta  > - (n-1)$, there are four  extreme points of $\mu_k$ taken at $k=0,n$ and locations near $\frac{(n-1) \pm \sqrt{n-1+2\bar\beta}}{2}$. Taking both situations and all extreme points into consideration, it is clear that the optimal $\bar\beta$ satisfies $2\bar\beta  > - (n-1)$ and is taken at a point where the absolute values of the four extreme points are the same. As it is difficult to calculate the exact value of the optimal coefficient, we first conveniently assume $n-1$ is divisible by $4$, and take $\bar\beta = (n-1)^2/8 - (n-1)/2$. In this case, the extreme points are taken at $k=0,\frac{n-1}{4},\frac{3(n-1)}{4},n$, and the corresponding eigenvalues of $A(0,\avg(\beta),0)$ are $\mu_k = \frac{(n+7)n(n-1)}{24},-\frac{(n+1)(n-1)(n-3)}{24},\frac{(n+1)(n-1)(n-3)}{24}$, $-\frac{(n+7)n(n-1)}{24}$. It implies 
\begin{equation}
\frac{(n+1)(n-1)(n-3)}{24} \leq \frac{(\Delta G_\eff)_\opt}{2} 
= \min_{\bar\beta \in \bR} \bigg\|\sum_{i<j<k} Z_iZ_jZ_k - \bar\beta \sum_{i} Z_i \bigg\| 
\leq \frac{(n+7)n(n-1)}{24}.
\end{equation} 
For general $n$, we have 
\begin{align}
{(\Delta G_\eff)_\opt} = \frac{n^3}{12} + O(n^2). 
\end{align}

\section{Proofs of Properties of Shortened RM Codes}\label{app:A}

The classical RM codes have numerous interesting properties. In \secref{subsec:QRM}, we outlined three key properties of the shortened RM code that are relevant to our discussion of QRM codes. Property~(1), i.e., $\overline{RM}(1,m)\subsetneq \overline{RM}(m-2,m)$, is a simple observation since the rows of $\overline{\sfG}(1,m)$ are the same as the first $m$ rows of $\overline{\sfG}(m-2,m)$. Here, we prove Property~(2) and (3).

\begin{manualtheorem}{2}
The vector $\bm{1}_{2^m-1}$ and those corresponding to the rows of $\overline{\sfG}(m-2,m)$ form a basis for $\overline{RM}(1,m)^\perp$.    
\end{manualtheorem}

\begin{proof}
We know that $RM(1,m)^\perp = RM(m-2,m)$~\cite{macwilliams1977theory}. For the shortened code, we first delete the all-1 row and then delete the all-0 column. Hence, it immediately follows that $\overline{RM}(m-2,m) \subseteq \overline{RM}(1,m)^\perp$. Recall that $\dim \overline{RM}(1,m)=m=\binom{m}{m-1}$ and $\dim\overline{RM}(m-2,m) =\binom{m}{1}+\binom{m}{2}+\cdots+\binom{m}{m-2}$. Then, $\dim \overline{RM}(1,m)+\dim \overline{RM}(m-2,m)=2^m-2$, which is 1 less than the dimension of the column space. Hence, $\overline{RM}(1,m)^\perp$ has one extra basis vector that is not in the row space of $\overline{RM}(m-2,m)$. We claim that $\bm{1}_{2^m-1}$ is the additional basis vector for $\overline{RM}(1,m)^\perp$. Since the rows of $\overline{\sfG}(1,m)$ are of even weight, it is clear that $\bm{1}_{2^m-1}$ is orthogonal to each row vector, and so $\bm{1}_{2^m-1}\in \overline{RM}(1,m)^\perp$. It remains to show that $\bm{1}_{2^m-1}\not \in \overline{RM}(m-2,m)$.

For a contradiction, assume that $\bm{1}_{2^m-1}\in \overline{RM}(m-2,m)$. Then there are rows $\bm{v}_{i_1}^T, \bm{v}_{i_2}^T, \ldots \bm{v}_{i_p}^T$ of $\overline{\sfG}(m-2,m)$ such that $\bm{v}_{i_1}^T+\bm{v}_{i_2}^T+\cdots+ \bm{v}_{i_p}^T = \bm{1}_{2^m-1}^T$. It follows that $(0|\bm{v}_{i_1}^T)^T+(0|\bm{v}_{i_2}^T)^T +\cdots+ (0|\bm{v}_{i_p}^T)^T = (0|\bm{1_{2^m-1}}^T)^T\in RM(m-2,m)$. Adding $\bm{1}_{2^m}\in RM(m-2,m)$ to this equation, we obtain $\bm{1}_{2^m}+(0|\bm{v}_{i_1}^T)^T+(0|\bm{v}_{i_2}^T)^T +\cdots+ (0|\bm{v}_{i_p}^T)^T = (1|\bm{0}_{2^m-1}^T)^T$, which is a weight-1 vector. However, $RM(m-2,m)$ has a minimum weight of $4$ (The minimum weight of $RM(r,m)$ is $2^{m-r}$~\cite{macwilliams1977theory}). So, we have a contradiction; hence, property~(2) is proven. 
\end{proof}

\begin{manualtheorem}{3}
The number of codewords in $\overline{RM}(1,m)^\perp$ with weight $(2^m-4)$ is $\frac{1}{6}(4^m-3\cdot 2^m+2)$.
\end{manualtheorem}

\begin{proof}
    First, we consider the weight distribution, i.e., the number of codewords of each weight, of $RM(1,m)$. From that, we will obtain the weight distribution of $\overline{RM}(1,m)$, which, in turn, will give us the weight distribution of $\overline{RM}(1,m)^\perp$.

    Let $A_k$ denote the number of codewords of weight $k$ in $RM(1,m)$. It can be shown that $A_0=1$, $A_{2^m}=1$, $A_{2^{m-1}}=2^{m+1}-2$, and $A_k=0$ for all $k\neq 0,1,2^{m-1}$~\cite{macwilliams1977theory}. Note that adding the vector $\bm{1}_{2^m}^T$ to any weight-$(2^{m-1})$ vector of the form $(0|\bm{v})^T\in RM(1,m)$ results in another weight-$(2^{m-1})$ vector of the form $(1|\bm{v}')^T$, where $\bm{v}$ and $\bm{v}'$ are $(2^m-1)$-tuple row vectors. So, there is a one-to-one correspondence (related by the addition of $\bm{1}_{2^m}$) between the weight-$2^{m-1}$ vectors with the first entry 0 and the weight-$2^{m-1}$ vectors with the first entry 1. The shortening process gets rid of the all-1 row vector from the generator matrix. Therefore, the number of weight-$2^{m-1}$ vectors in $\overline{RM}(1,m)$ is half of that in $RM(1,m)$. Denoting $\overline{A}_k$ as the number of weight-$k$ codewords in $\overline{RM}(1,m)$, we get $\overline{A}_0=1$, $\overline{A}_{2^{m-1}}=2^m-1$ and $\overline{A}_k=0$ for all $k\neq 0, 2^{m-1}$. 

    The \textit{enumerator polynomial} $W(C;x,y)$ succinctly represents the weight distribution of a classical code $C$, with the coefficient of $x^ly^{n-l}$ representing the number of codewords of weight $l$. Based on the values of $\overline{A}_k$, the enumerator polynomial of $\overline{RM}(1,m)$ is 
    \begin{align}
        W(\overline{RM}(1,m);x,y)=y^{2^m-1}+(2^m-1)x^{2^{m-1}}y^{2^{m-1}-1}.
        \label{eqn: enumshortRM(1,m)}
    \end{align}
    One can directly obtain the weight distribution of the dual code $C^\perp$ from the enumerator polynomial of $C$ using the \textit{Macwilliams identity} 
    \begin{equation}
    W(C^\perp;x,y)=\frac{1}{|C|}W(C; y-x,x+y),     
    \end{equation} 
    where $|C|$ is the number of codewords in $C$~\cite{macwilliams1977theory}. Since $|\overline{RM}(1,m)|=2^m$, applying the Macwilliams identity to $W(\overline{RM}(1,m);x,y)$ yields
    \begin{align}
        W(\overline{RM}(1,m)^\perp;x,y) =\frac{1}{2^m}(x+y)^{2^m-1} +\frac{(2^m-1)}{2^m} (y-x)^{2^{m-1}} (x+y)^{2^{m-1}-1}.
        \label{eqn: enumshortRM(1,m)perp}
    \end{align}
    We want to find the coefficient of $x^{2^m-4}y^3$ on the right side of \eqref{eqn: enumshortRM(1,m)perp}. Through a simple combinatorial calculation, the coefficient is found to be $\frac{1}{2^m}\binom{2^m-1}{3}+\frac{2^m-1}{2^m}\big[-\binom{2^{m-1}}{3}\binom{2^{m-1}-1}{0}+ \binom{2^{m-1}}{2}\binom{2^{m-1}-1}{1}- \binom{2^{m-1}}{1}\binom{2^{m-1}-1}{2}+ \binom{2^{m-1}}{0}\binom{2^{m-1}-1}{3}\big]$, which simplifies to $\frac{1}{6}(4^m-3\cdot 2^m+2)$. So, the number of weight-$(2^m-4)$ codewords in $\overline{RM}(1,m)^\perp$ is $\frac{1}{6}(4^m-3\cdot 2^m+2)$.
\end{proof}

\section{Graph Theory Lemma}

In \secref{sec: no go}, we used several graph-theoretic tools to establish the no-go results. Specifically, in \secref{subsec: cubic scaling}, we claimed that any graph with a quadratic number of edges has a connected component with a linear number of vertices and a quadratic number of edges. We now prove this claim in the following lemma.

\begin{lemma}
    Consider a family of graphs $\{\frakG^{(n)}\}_n$ where $n$ belongs to an infinitely large set of positive integers. 
    Assume $\frakG^{(n)}$ has $\Theta(n)$ vertices and $\Theta(n^2)$ edges. Then there is a family of connected components of $\frakG^{(n)}$ with $\Theta(n)$ vertices and $\Theta(n^2)$ edges.
    \label{lemma: graphlemma}
\end{lemma}

\begin{proof}
    For a sufficiently large $n$, let $\frakG = \frakG^{(n)}$ be a graph with $\geq an^2$ edges and $\geq bn$ vertices for some $a,b > 0$. Let the connected components of $\frakG$ be $\frakG_1$, $\frakG_2$, $\ldots, \frakG_k$, and for all $i$, the number of vertices and edges in $\frakG_i$ be $\bar v_i$ and $\bar e_i$, respectively. We index the connected components in descending order of the number of vertices. Let $\bar e=\max\{ \bar e_i\}_{i=1}^k$. Suppose that the first $l$ connected components have at least $\bar v=\lceil \sqrt{2\bar e}+1 \rceil$ vertices, i.e., enough vertices to make a complete graph with $\bar e$ edges. 

    Consider another graph $\frakF$ with the same number of vertices as $\frakG$ constructed in the following steps. Note that we will first partition vertices into different sets in steps (1) and (2) and then add edges within different sets in step (3). 
    \begin{enumerate}[wide, labelwidth=!,itemindent=!,labelindent=0pt, leftmargin=0em, label={(\arabic*)}, parsep=0pt]
        \item Let $V_1,V_2,\ldots,V_l$ be subsets of vertices from $\frakG_1,\ldots,\frakG_l$, where for each $\frakG_i$ ($1\leq i \leq l$), we only select $\bar v$ vertices to place in $V_i$ and place the rest of all vertices of $\mathfrak{G}$ in the set $V^R$. 
        \item Partition the set of vertices $V^R$ into subsets $V_{l+1},$ $V_{l+2}, \ldots V_{k'}, V^{R'}$ such that each $V_i$ ($1\leq i \leq k'$) has $\bar v$ vertices and the remainder set $V^{R'}$ has $\bar w < \bar v$ vertices. 
        \item Add edges between all vertices within each subsets of vertices $V_1,V_2, \ldots, V_{k'}, V^{R'}$ to make them complete subgraphs $\frakF_1,\frakF_2, \ldots, \frakF_{k'}, \frakR$, which compose graph $\frakF$. 
    \end{enumerate}
Connected components $\frakF_1, \frakF_2, \ldots, \frakF_{k'}$ of $\frakF$ have at least $\bar e$ edges each, and $\frakR$ has $\bar w(\bar w-1)/2$ edges. So, the number of edges in $\frakF$ is at least $\left(\frac{b n - \bar w}{\bar v} \right)\bar e+\frac{\bar w(\bar w-1)}{2}$. From the construction, $\frakF$ has at least as many edges as $\frakG$ does. Therefore,
    \begin{gather}
        \left(\frac{bn-\bar w}{\sqrt{2\bar e}}\right)\bar e+\frac{\bar w(\bar w-1)}{2}\geq a n^2, \nonumber\\ \Rightarrow~~~ bn\sqrt{\bar e/2} + \sqrt{\bar e/2} + 1 \geq an^2,
    \end{gather}
    where we set $\bar w = \sqrt{2\bar e} + 2$ in the second step which increases the left-hand side so that the inequality still holds. Solving for $\bar e$, we get $\bar e=\Omega(n^2)$. A graph with $\Omega(n^2)$ edges must have $\Omega(n)$ vertices. However, the number of vertices in any connected component of $\frakG$ is $O(n)$, which implies that $\bar e=O(n^2)$. Therefore, $\bar e=\Theta(n^2)$ and the number of vertices in a component with $\bar e$ edges is $\Theta(n)$. Hence, $\frakG$ has a connected component with $\Theta(n)$ vertices and $\Theta(n^2)$ edges.
\end{proof}

\end{document}